\documentclass{amsart}

\usepackage{amsmath}
\usepackage{amsfonts}
\usepackage{bm}
\usepackage{amsthm}
\usepackage{graphicx}

\newtheorem{theorem}{Theorem}[section]
\newtheorem{remark}[theorem]{Remark}
\newtheorem{assumption}[theorem]{Assumption}
\newtheorem{lemma}[theorem]{Lemma}
\newtheorem{proposition}[theorem]{Proposition}

\newcommand{\va}{\bm{a}}
\newcommand{\vx}{\bm{x}}
\newcommand{\vy}{\bm{y}}
\newcommand{\vz}{\bm{z}}
\newcommand{\ve}{\bm{e}}
\newcommand{\vb}{\bm{b}}
\newcommand{\vv}{\bm{v}}
\newcommand{\vu}{\bm{u}}
\newcommand{\vw}{\bm{w}}

\newcommand{\vf}{\bm{f}}

\newcommand{\vd}{\bm{d}}

\newcommand{\vp}{\bm{p}}
\newcommand{\vg}{\bm{g}}

\newcommand{\sA}{\mathcal{A}}
\newcommand{\sB}{\mathcal{B}}
\newcommand{\sW}{\mathcal{W}}

\newcommand{\mM}{\mathsf{M}}

\newcommand{\mD}{\mathsf{D}}
\newcommand{\mC}{\mathsf{C}}

\newcommand{\mI}{\mathsf{I}}
\newcommand{\mF}{\mathsf{F}}
\newcommand{\mR}{\mathsf{R}}

\newcommand{\tr}{\mathsf{T}}

\newcommand{\KM}{\Gamma_{\text{KM}}}
\newcommand{\real}{\mathbb{R}}

\newcommand{\complex}{\mathbb{C}}
\newcommand{\zed}{\mathbb{Z}}
\newcommand{\bigO}{\mathcal{O}}
\newcommand{\CE}{\Psi}
\newcommand{\C}{\psi}
\renewcommand{\Re}{\text{Re}}
\renewcommand{\Im}{\text{Im}}
\newcommand{\ii}{\imath}
\DeclareMathOperator\cov{Cov}
\DeclareMathOperator\var{Var}
\DeclareMathOperator\sinc{sinc}
\DeclareMathOperator\nullspace{null}
\DeclareMathOperator\linspan{span}
\DeclareMathOperator\rank{rank}

\renewcommand{\hat}{\widehat}

\title{Imaging with power controlled source pairs}
\author{Patrick Bardsley}
\email{bardsley@math.utah.edu}
\address{Mathematics Department, University of Utah, 155 S 1400 E RM 233, Salt Lake City
UT 84112-0090.}

\author{Fernando Guevara~Vasquez}
\email{fguevara@math.utah.edu}
\address{Mathematics Department, University of Utah, 155 S 1400 E RM 233, Salt Lake City
UT 84112-0090.}

\begin{document}

\maketitle

%%%%%%%%%%%%%%%%%%%%%%%%%%%%%%%%%%%%%%%%%%%%%%%%%%%%%%%%%%%%%%%%%%%%%%%%
%% Abstract
\begin{abstract}
Scatterers in a homogeneous medium are imaged by probing the medium with
two point sources of waves modulated by correlated signals and by
measuring only intensities at one single receiver. For appropriately
chosen source pairs, we show that full waveform array measurements can
be recovered from such intensity measurements by solving a linear least
squares problem. The least squares solution can be used to image with
Kirchhoff migration, even if the solution is determined only up to a known
one-dimensional nullspace. The same imaging strategy can be used
when the medium is probed with point sources driven by correlated
Gaussian processes and autocorrelations are measured at a single
location. Since autocorrelations are robust to noise, this can be used
for imaging when the probing wave is drowned in background noise.

\smallskip
\noindent \textbf{Keywords.} Intensity-only imaging, travel-time migration, noise
sources, autocorrelation.

\smallskip
\noindent \textbf{AMS Subject Classifications.} 78A45, 78A46, 35R30

\end{abstract}

%%%%%%%%%%%%%%%%%%%%%%%%%%%%%%%%%%%%%%%%%%%%%%%%%%%%%%%%%%%%%%%%%%%%%%%%
%% Keywords, AMS classification, Page Headings

%%%%%%%%%%%%%%%%%%%%%%%%%%%%%%%%%%%%%%%%%%%%%%%%%%%%%%%%%%%%%%%%%%%%%%%%%
%%% Introduction
\section{Introduction}\label{sec:intro}

Scatterers in a homogeneous medium can be imaged by probing the medium
with a wave emanating from a point source, and recording the reflected
waves at one or more receivers. An image of the scatterers can be
generated by repeating this experiment while varying the position of the
source and/or receiver and using classic methods such as the Kirchhoff (travel
time) migration (see e.g. \cite{Bleistein:2001:MMI}) or MUSIC (see e.g.
\cite{Cheney:2001:LSM}). We are concerned here with the case where only
{\em intensity measurements} can be made at the receiver; destroying
phase information that migration and MUSIC need to image. Intensity
measurements occur e.g. when the response time of the receiver is larger
than the typical wave period or when it is more cost effective to
measure intensities than the full waveform.  This is typical in e.g.
optical coherence tomography \cite{Schmitt:1997:OCM,Schmitt:1999:OCT} and radar
imaging \cite{Cheney:2009:FRI}. Another situation is when the wave
sources are stochastic and the measurements consist of correlations of
the signal recorded at different points
\cite{Schuster:2004:ISI,Garnier:2009:PSI}. In the special case of
autocorrelations (i.e. correlating the signal with itself), the
Wiener-Khinchin theorem guarantees we are measuring power spectra (see
e.g. \cite{Ishimaru:1997:WPS}), another form of intensity measurements. 

The setup we analyze consists of an array of sources and one single
receiver that can only record power spectra, i.e. the intensity of the
signal at certain frequency samples. A crucial assumption for our method
is that we can use {\em source pairs}, meaning we can send {\em
correlated signals} from {\em two different locations}. Thus we allow
for known delays or attenuations between the signals in a source pair.
In acoustics, one way of achieving this would be to drive two
transducers in an array with the same signal. With light, one could use
an incoherent plane wave with wavefronts parallel to a configurable
mask. The mask lets light through one or two small holes, whose
locations can be controlled. 

Our method can be used for imaging from both measurements of intensities
(\S\ref{sec:intensities}) and autocorrelations (\S\ref{sec:correlation}).

%%%%%%%%%%%%%%%%%%%%%%%%%%%%%%%%%%%%%%%%%%%%%%%%%%%%%%%%%%%%%%%%%%%%%%%%
\subsection{Intensity only measurements}
\label{sec:intensities}
One way to deal with intensity measurements is {\em phase
retrieval}, i.e. first recovering the phases from intensity
measurements, and using this reconstructed field to image. In
diffraction tomography, intensity measurements at two different planes
can be used to recover phases
\cite{Gerchberg:1972:GSA,Teague:1983:DPR,Gbur:2004:ICS}. If additional
information is known (e.g. the support
of the scatterer), intensities at one single plane can be used
\cite{Fienup:1987:RCV,Maleki:1992:TRO,Maleki:1993:PRI}.  Total or partial knowledge
of the incident field can also be exploited to image from intensities at one single
plane \cite{Crocco:2007:FNL}.

Chai et al. \cite{Chai:2011:AII} take a compressed sensing approach to image a few
point scatterers exactly. With knowledge of the incident field, the location
of the scatterers can be resolved in both range and cross-range with
monochromatic measurements. The same ideas can even be used to 
deal with multiple scattering \cite{Chai:2014:ISL}. Novikov et al.
\cite{Novikov:2014:ISI} use the polarization identity $4\Re(\vu^*\vv)
= \|\vu+\vv\| ^2 - \|\vu-\vv\|^2$,
$\vu,\vv \in \complex^N$, and linear combinations of single source
experiments to recover dot products of two single source experiments
from intensity data. MUSIC can then be used to image with this quadratic
functional of the full waveform data.

Here we do phase retrieval assuming knowledge of the intensity of
the incident field. Our illumination strategy using source pairs does
not require direct manipulation of phases or addition/subtraction of
wave fields. We reduce the recovery of the total field to a linear
system with a one-dimensional nullspace which we can write explicitly in
terms of the incident field. There is one (very sparse) linear system per frequency
sample to solve, and the linear system has size comparable to twice the
number of source positions. Intuitively we are recovering a field in
$\complex^N$ from $2N$ (or more) real measurements. We show that
vectors in the one-dimensional nullspace do not affect Kirchhoff
migration. Hence we can use, without modification, Kirchhoff migration and
its standard range and cross-range resolution estimates
(see e.g. \cite{Bleistein:2001:MMI}).

%%%%%%%%%%%%%%%%%%%%%%%%%%%%%%%%%%%%%%%%%%%%%%%%%%%%%%%%%%%%%%%%%%%%%%%%
\subsection{Correlation based methods}
\label{sec:correlation}

In seismic imaging, correlations of traces (or recordings) at many
receivers have been used to image the earth's subsurface, especially
when the wave sources and their locations are not well known
\cite{Schuster:1996:RLC,Schuster:2004:ISI,Schuster:2009:SI}. The idea is
that correlations of the signals at two different locations contain
information about the Green's function between the two locations, and
this information can be exploited to image the medium and any
scatterers. This principle can even be exploited to do opportunistic
imaging with ambient noise
\cite{Garnier:2009:PSI,Garnier:2010:RAI,Garnier:2015:SNR}.
Cross-correlations can also be used to image scatterers in a random
medium 
\cite{Borcea:2005:IAI,Garnier:2012:CBV,Garnier:2014:RSV}. In radar imaging, the
measurements are in fact correlations \cite{Cheney:2009:FRI}, and so
even stochastic processes can be used instead of deterministic signals
\cite{Tarchi:2010:SARNR,Vela:2012:NRT}.

The method we present here can also be used to image scatterers using
autocorrelations. We show it is possible to form an image by exploiting angular diversity in
source pairs instead of cross-correlations among different receivers.
Just as in the intensity measurements case, we are able to recover (up to a one
dimensional nullspace) full waveform array measurements. One advantage of using
autocorrelations instead of cross-correlations is that the data acquisition at
the (single) receiver is simpler.  The drawback is that our illumination
strategy requires to illuminate with pairs of sources, but also with each of
the sources in a pair on its own. To get the same full waveform data as an array with $N$
sources, we need at least $3N$ different experiments. Another advantage of
using autocorrelations is that the measurements are extremely robust to noise.
As an example, our numerical experiments show that it is possible to image
scatterers with an array that is sending noise from all possible source
positions; the only assumption about the noise being that all the
sources are independent stochastic processes except for two correlated
sources whose positions we can control. Because of this robustness, it
may be possible to use our imaging method in situations where the medium
is to be probed in a non-intrusive way, i.e. active imaging with waves
that are of the same magnitude as the ambient noise.

%%%%%%%%%%%%%%%%%%%%%%%%%%%%%%%%%%%%%%%%%%%%%%%%%%%%%%%%%%%%%%%%%%%%%%%%
\subsection{Contents}

The particular physical setup we consider is described in
\S\ref{sec:arrayimaging}. The illumination strategy with source pairs is explained
in \S\ref{sec:arvrecovery}, which leads to a phase retrieval problem
that can be formulated as a linear system (\S\ref{sec:invertible}).  The
least squares solution to the linear system is then used as data for
imaging with Kirchhoff migration, and we show that this gives
essentially the same images as full waveform data
(\S\ref{sec:migration}). The extension to stochastic source pairs is
given in \S\ref{sec:stochillum}. Then we show that our method is robust
to additive noise when using autocorrelations (\S\ref{sec:snr}).
Numerical experiments illustrating our method are given in
\S\ref{sec:numerics} and we conclude with a discussion in
\S\ref{sec:discussion}.

%%%%%%%%%%%%%%%%%%%%%%%%%%%%%%%%%%%%%%%%%%%%%%%%%%%%%%%%%%%%%%%%%%%%%%%%%
%%% Setup
%%%%%%%%%%%%%%%%%%%%%%%%%%%%%%%%%%%%%%%%%%%%%%%%%%%%%%%%%%%%%%%%%%%%%%%%
\section{Array imaging for full waveform measurements}\label{sec:arrayimaging}

Here we introduce the experimental setup we consider (\S\ref{sec:exp})
and briefly recall the classic Kirchhoff migration imaging method
(\S\ref{sec:km}).

%%%%%%%%%%%%%%%%%%%%%%%%%%%%%%%%%%%%%%%%%%%%%%%%%%%%%%%%%%%%%%%%%%%%%%%%
\subsection{Experimental setup}\label{sec:exp}
The physical setup is illustrated in figure~\ref{fig:arraysetup}. We
probe a homogeneous medium with waves originating from $N$ point sources
with locations $\vec{\vx}_s \in \sA$, $s = 1,2,\ldots,N$. For simplicity
we consider a linear array in 2D or a square array in 3D, i.e. $\sA =
[-a/2,a/2]^{d-1} \times \{0\}$, where $d=2$ or $3$ is the dimension.
Our imaging strategy imposes only mild restrictions on the source
positions, so other array shapes may be considered. Waves are recorded
at a {\em single} known receiver location $\vec{\vx}_r$. 

The total field generated by the array (or incident field) can be
written as
\begin{equation}
 \hat{u}_{\text{inc}}(\vec{\vx},\omega) = \vg_0(\vec{\vx},\omega)^\tr \vf(\omega),
\end{equation}
where
\begin{equation}\label{eq:homoggreen}
\vg_0(\vec{\vx},\omega) = \begin{bmatrix}
\hat{G}_0(\vec{\vx},\vec{\vx}_1,\omega),
\hat{G}_0(\vec{\vx},\vec{\vx}_2,\omega),
\ldots,
\hat{G}_0(\vec{\vx},\vec{\vx}_N,\omega)
\end{bmatrix}^\tr \in \complex^N,
\end{equation}
and the source driving signals are $\vf(\omega) =
[\hat{f}_1(\omega), \hat{f}_2(\omega), \ldots,
\hat{f}_N(\omega)]^\tr$. Since we assume waves propagate through a
homogeneous medium, we used the outgoing free space Green function,
\begin{equation}
 \hat{G}_0(\vec{\vx},\vec{\vy},\omega) = 
 \begin{cases}
  \frac{\ii}{4} H_0^{(1)}(k|\vec{\vx} - \vec{\vy}|),  & \text{for
  $d=2$},\\
  \displaystyle\frac{\exp[\ii k|\vec{\vx} - \vec{\vy}|]}{4\pi |\vec{\vx} -
 \vec{\vy}|}, & \text{for $d=3$.}
 \end{cases}
 \label{eq:green}
\end{equation}
Here $H_0^{(1)}$ is the zeroth order Hankel function of the first kind,
$k = \omega/c_0$ is the wavenumber, $\omega$ is the angular frequency
and $c_0$ is a known constant background wave speed. For functions of
time, the Fourier transform convention we use is
\begin{equation}\label{eq:FT}
\hat{f}(\omega) = \int f(t)e^{\ii\omega t} dt,~\text{and}~
f(t) = \frac{1}{2\pi}\int \hat{f}(\omega)e^{-\ii\omega
t}d\omega,~\text{where $f\in L^2(\real)$.}
\end{equation}

The scatterers we want to image are represented by a compactly supported
reflectivity function $\rho(\vec{\vx})$. Under the weak scattering
assumption (i.e. $\rho \ll 1$), we can use the Born
approximation to the total field at the receiver
\begin{equation}\label{eq:fullmeas}
\hat{u}(\vec{\vx}_r,\omega) =
(\vg_0+\vp)^\tr\vf,
\end{equation}
where the array response vector is
\begin{equation}\label{eq:arv}
\vp(\vec{\vx},\omega) = k^2\int d\vec{\vy}
\rho(\vec{\vy})\hat{G}_0(\vec{\vx},\vec{\vy},\omega)\vg_0(\vec{\vy},\omega).
\end{equation}

\begin{figure}
\centering
\includegraphics[scale=0.7]{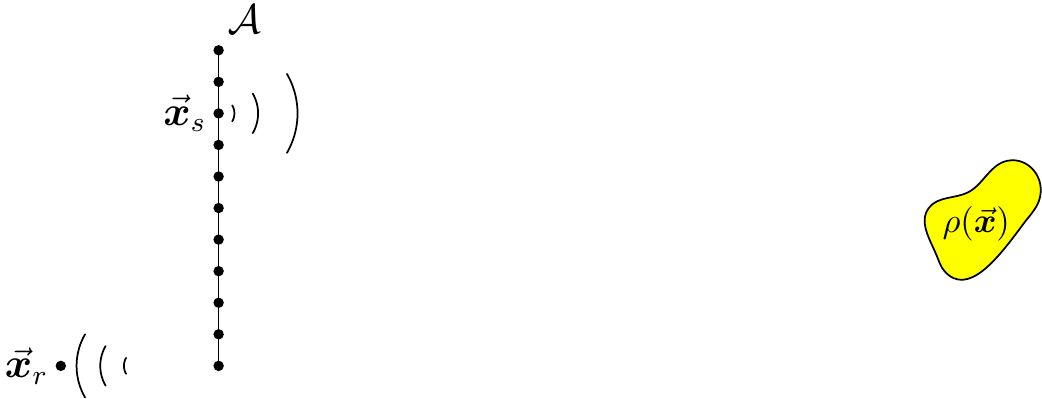}
\caption{Physical setup for array imaging with an array $\sA$ of sources
$\vec{\vx}_s$ and a single receiver $\vec{\vx}_r$. The scatterer is
represented by a compactly supported reflectivity function $\rho(\vec{\vx})$.}
\label{fig:arraysetup}
\end{figure}

%%%%%%%%%%%%%%%%%%%%%%%%%%%%%%%%%%%%%%%%%%%%%%%%%%%%%%%%%%%%%%%%%%%%%%%%
\subsection{Kirchhoff migration}\label{sec:km}

By e.g. using illuminations $\vf(\omega) = \ve_i$, $i=1,\ldots,N$ corresponding to the
canonical basis vectors, it is possible to obtain the
array response vector $\vp(\vec{\vx}_r,\omega)$ from the
measurements \eqref{eq:fullmeas}. The scatterers can then be imaged using
the Kirchhoff migration functional (see e.g. 
\cite{Bleistein:2001:MMI}) which for a single frequency $\omega$ is
\begin{equation}\label{eq:kirchoff}
\KM[\vp,\omega](\vec{\vy}) = 
\overline{\hat{G}}_0(\vec{\vy},\vec{\vx}_r,\omega)\vg_0(\vec{\vy},\omega)^*\vp(\vec{\vx}_r,\omega),
\end{equation}
where $\vec{\vy}$ represents a point in the image. This image has a
Rayleigh or cross-range (i.e. in the direction parallel
to the array) resolution of $\lambda L /
a$, where $L$ is the distance from the array to the scatterer
(see e.g. \cite{Bleistein:2001:MMI}). To get range (i.e. in the direction
perpendicular to the array) resolution we need to integrate
$\KM[\vp,\omega](\vec{\vy})$ for frequencies $\omega$ in some
frequency band $\sB =
[-\omega_{max},-\omega_{min}] \cup [\omega_{min},\omega_{max}]$, the
same frequency band of the signals $\vf(\omega)$.  The range
resolution is then $c_0/(\omega_{max} - \omega_{min})$ (see e.g.
\cite{Bleistein:2001:MMI}). We discuss this imaging functional further
in section~\ref{sec:migration}.

%%%%%%%%%%%%%%%%%%%%%%%%%%%%%%%%%%%%%%%%%%%%%%%%%%%%%%%%%%%%%%%%%%%%%%%%%
%%% ARV Recovery
%%%%%%%%%%%%%%%%%%%%%%%%%%%%%%%%%%%%%%%%%%%%%%%%%%%%%%%%%%%%%%%%%%%%%%%%
\section{Intensity only measurements}
\label{sec:arvrecovery} 
We start in \S\ref{sec:detillum} by describing a source pair
illumination strategy for intensity measurements of the total field
$|\hat{u}(\vec{\vx}_r,\omega)|^2$.  With this strategy, the problem of
recovering the array response vector $\vp$ can be formulated as a
linear system (\S\ref{sec:prsys}).  
%%%%%%%%%%%%%%%%%%%%%%%%%%%%%%%%%%%%%%%%%%%%%%%%%%%%%%%%%%%%%%%%%%%%%%%%
\subsection{Illumination strategy}\label{sec:detillum}
The data we use comes from probing the medium with $N_p$ source pairs
that are sending signals with known power and phase difference. Since
the number of distinct source pairs out of an array with $N$ sources is
$N(N-1)/2$ we must have $N_p \leq N(N-1)/2$. We assume the power and
phase differences remain the same for all $N_p$ illuminations.  The case
where these quantities depend on the source pair is left for future
studies. To be more precise, the illumination corresponding to the
$m-$th source pair $(i(m),j(m)) \in \{1,\ldots,N\}^2$ is
\begin{equation}
\vf_m(\omega) = 
\mF_m \begin{bmatrix} \alpha(\omega) \\ \beta(\omega) \end{bmatrix},
~\text{where}~\mF_m = [\ve_{i(m)}, \ve_{j(m)}] \in \real^{N\times 2}.
\label{eq:Fm}
\end{equation}
We emphasize that only $|\alpha|^2$, $|\beta|^2$
and the phase difference $\phi(\omega) \equiv \arg(\overline{\alpha}
\beta)$ is assumed to be known for the signals $\alpha$ and $\beta$. A
particular case is when the same signal is sent from the source pair,
i.e. $\beta = \alpha$ and $\phi(\omega) = 0$.

The intensity of the field $u_m$ arising from the source pair
illumination $\vf_m$ is
\begin{equation}\label{eq:det_inhomog_meas}
|\hat{u}_m(\vec{\vx}_r,\omega)|^2 =
\overline{\vg^\tr \vf_m} \vf_m^\tr \vg = \vg^* \mF_m \begin{bmatrix}
|\alpha|^2 & \overline{\alpha}\beta \\ \overline{\beta} \alpha &
|\beta|^2 \end{bmatrix} \mF_m^\tr \vg,
\end{equation}
where we used $\vg = \vg_0  + \vp$. Note that since
$\overline{\alpha} \beta = |\alpha| |\beta| e^{\ii\phi}$, the inner $2\times 2$
Hermitian matrix is uniquely determined by the magnitudes of $\alpha$
and $\beta$ and their phase difference $\phi$. By using the single
source reference illumination $\ve_{i}$ we additionally measure
\begin{equation}\label{eq:det_homog_meas}
|\hat{u}_{i}^0(\vec{\vx}_r,\omega)|^2 =
\vg^*\ve_{i}\ve_{i}^\tr\vg, \qquad\text{ for $i=1,\ldots,N.$}
\end{equation} 
The data we exploit to recover $\vp$ is obtained by subtracting the
appropriate reference illuminations \eqref{eq:det_homog_meas} from
\eqref{eq:det_inhomog_meas}, that is
\[
\begin{aligned}
d_m(\vec{\vx}_r,\omega) &= |\hat{u}_m|^2 -
|\alpha|^2|\hat{u}_{i(m)}^0|^2 -
|\beta|^2|\hat{u}_{j(m)}^0|^2\\
&=
\vg^* \mF_m \begin{bmatrix}
0 & \overline{\alpha}\beta \\ \overline{\beta} \alpha & 0 \end{bmatrix}
\mF_m^\tr \vg.
\end{aligned}
\]

%%%%%%%%%%%%%%%%%%%%%%%%%%%%%%%%%%%%%%%%%%%%%%%%%%%%%%%%%%%%%%%%%%%%%%%%
\subsection{Phase retrieval problem as a linear system}
\label{sec:prsys}
By recalling that $\vg = \vg_0 + \vp$, the measurements $d_m$ 
are
\[
\begin{aligned}
d_m(\vec{\vx}_r,\omega) =
(\vg_0+\vp)^*\mF_m \begin{bmatrix}
0 & \overline{\alpha}\beta \\ \overline{\beta} \alpha &
0 \end{bmatrix} \mF_m^\tr(\vg_0+\vp).
\end{aligned}
\]
To make the following expressions concise, we denote by $\mD$ the
Hermitian matrix
\begin{equation}
\mD = \begin{bmatrix}
0 & \overline{\alpha}\beta \\ \overline{\beta} \alpha &
0 \end{bmatrix}.
\label{eq:mD}
\end{equation}
By the weak scattering assumption, we may neglect the quadratic terms in
$\vp$ and collect all measurements for $m=1,\ldots,N_p$ as a single
vector $\vd \in \real^{N_p}$:
\begin{equation}\label{eq:detdata}
\begin{aligned}
\begin{bmatrix} 
d_1(\vec{\vx}_r,\omega)\\
d_2(\vec{\vx}_r,\omega)\\
\vdots\\
d_{N_p}(\vec{\vx}_r,\omega)\\
\end{bmatrix}
\approx
\vd(\vec{\vx}_r,\omega)
&= \Re\begin{pmatrix}
\begin{bmatrix}
\vg_0^*\mF_1 \mD \mF_1^\tr \\
\vg_0^*\mF_2 \mD \mF_2^\tr\\
\vdots\\
\vg_0^*\mF_{N_p} \mD \mF_{N_p}^\tr
\end{bmatrix}
(\vg_0+2\vp)
\end{pmatrix}\\
&= \mM(\vec{\vx}_r,\omega)
\begin{bmatrix}
\Re(\vg_0+2\vp)\\
\Im(\vg_0+2\vp)
\end{bmatrix},
\end{aligned}
\end{equation} 
where the $N_p\times 2N$ real matrix $\mM$ is given
by
\begin{equation}\label{eq:M}
\mM(\vec{\vx}_r,\omega)=\begin{bmatrix}
\Re(\vg_0^* \mF_1 \mD \mF_1^\tr) & -\Im(\vg_0^* \mF_1 \mD \mF_1^\tr) \\
\Re(\vg_0^* \mF_2 \mD \mF_2^\tr) & -\Im(\vg_0^* \mF_2 \mD \mF_2^\tr) \\
\vdots & \vdots\\
\Re(\vg_0^* \mF_{N_p} \mD \mF_{N_p}^\tr) & -\Im(\vg_0^* \mF_{N_p} \mD
\mF_{N_p}^\tr) \\
\end{bmatrix}.
\end{equation}

Note that by construction, the matrix $\mM$ has at most 4 non-zero
elements per row, and is thus a very sparse matrix for $N$ large.

%%%%%%%%%%%%%%%%%%%%%%%%%%%%%%%%%%%%%%%%%%%%%%%%%%%%%%%%%%%%%%%%%%%%%%%%
\section{Analysis of the phase retrieval linear system}\label{sec:invertible} 
We now address the question of whether there is enough information in
the measurements $\vd \in \real^{N_p}$ to recover the array response
vector $\vp \in \complex^{N}$. The main result of this section is
Theorem~\ref{thm:invertibility}, where we show that with appropriately
chosen pairs of sources, $\mM^\dagger \vd$ (i.e. the Moore Penrose
pseudoinverse of $\mM$ times $\vd$) gives $\vp$ up to a complex scalar
multiple of the vector $\vg_0$, which is known a priori.

Let us first consider the case where we take measurements using all
possible source pairs, i.e. that $N_p = N(N-1)/2$. Clearly, we need
$N\geq 5$ to guarantee that $N_p \geq 2 N$, i.e. that the matrix $\mM$
has more rows than columns and the system $\vd=\mM[\Re(\vg_0+2\vp)^\tr,
\Im(\vg_0+2\vp)^\tr]^\tr$ is overdetermined.

Instead of using all possible source pairs, we use the following
strategy which for $N \geq 5$, guarantees $N_p = 2N$.

{\bf Strategy to choose source pairs}:
\begin{enumerate}
 \item All 10 distinct source pairs between the source positions
 $\{1,\ldots,5\}$.
 \item For source position $s>5$, choose any two different source pairs of the form
 $(s,i)$ and $(s,j)$ where $i,j \in \{1,\ldots,5\}$.
\end{enumerate}
This strategy is illustrated in
figure~\ref{fig:strategy}. More source pairs can be
added without affecting the recoverability of $\vp$ 
(Theorem~\ref{thm:invertibility}). We now make the following assumption
on the first 5 source positions.
\begin{figure}
 \begin{center}
 \includegraphics[width=0.4\textwidth]{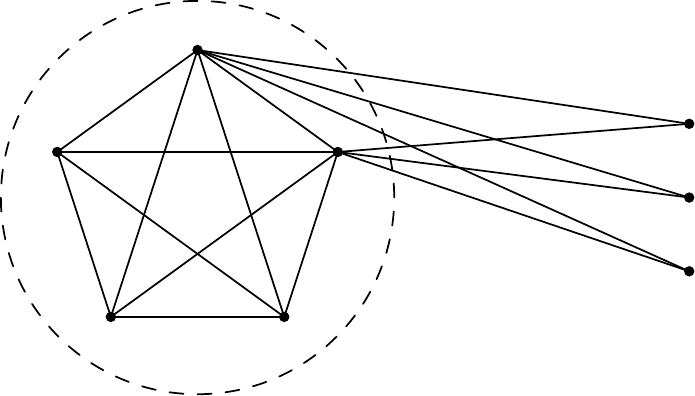}
 \end{center}
 \caption{An example illustrating the strategy to choose the source
 pairs for $N=8$ source positions. Each source position is represented
 by a node in the graph, and source pairs are represented by edges. The
 first 5 source positions are in the circle.}
 \label{fig:strategy}
\end{figure}
%%%%%%%%%%%%%%%%%%%%%%%%%%%%%%%%%%%%%%%%%%%%%%%%%% ASSUMPTION
\begin{assumption}\label{assump:rankcond} We assume the
receiver is located at a position $\vec{\vx}_r$ such that for
$i,j=1,\ldots,5$, the vector $\vg_0 \equiv \vg_0(\vec{\vx}_r,\omega)$
satisfies
\begin{equation}\label{eq:rankcond}
\Re\big(\vg_0\big)_i\neq0,~
\Im\big(\vg_0\big)_i\neq0, ~\text{and}~
\Re\big(\vg_0\big)_i\Im\big(\vg_0\big)_j
\neq
\Re\big(\vg_0\big)_j\Im\big(\vg_0\big)_i.
\end{equation}
Additionally for one pair $i,j\in[1,\ldots,5]$ we assume
\begin{equation}\label{eq:rankcond2}
\begin{aligned}
\cos(\phi)&\left(\Re\big(\vg_0\big)_i\Im\big(\vg_0\big)_j-\Re\big(\vg_0\big)_j\Im\big(\vg_0\big)_i\right)
\neq\\
&
\qquad\qquad-\sin(\phi)\left(\Re\big(\vg_0)_i\Re\big(\vg_0\big)_j+\Im\big(\vg_0\big)_i\Im\big(\vg_0\big)_j\right).
\end{aligned}
\end{equation}
\end{assumption}
This assumption is by no means necessary for the end result
(Theorem~\ref{thm:invertibility}) to hold, but it is sufficient. If
$d=3$, condition \eqref{eq:rankcond} is equivalent to the geometric
condition
\begin{equation}\label{eq:geomcond}
|\vec{\vx}_{i} - \vec{\vx}_r|  \notin \frac{\lambda}{4} \zed
~\text{and}~
|\vec{\vx}_{i} - \vec{\vx}_r| - |\vec{\vx}_{j} - \vec{\vx}_r| \notin
\frac{\lambda}{2} \zed ~\text{for all $i,j=1,\ldots,5$},
\end{equation}
while condition \eqref{eq:rankcond2} implies for one pair
$i,j\in[1,\ldots,5]$ that
\begin{equation}\label{eq:geomcond2}
|\vec{\vx}_r-\vec{\vx}_i|-|\vec{\vx}_r-\vec{\vx}_j| \notin
\frac{\lambda}{2}\zed-\frac{\lambda}{2\pi}\phi.
\end{equation}
Here the set $(\lambda/2)\zed$ is the set of all integer multiples of
$\lambda/2$, where $\lambda=2\pi c_0/\omega$ is the wavelength. In
$d=2$, conditions similar to \eqref{eq:geomcond} and
\eqref{eq:geomcond2} are sufficient when the sources and receivers are
far apart because of the Hankel function asymptotic 
\[
 H_0^{(1)} (t) = \sqrt{\frac{2}{\pi t}} \exp[\ii(t-(\pi/4))] (1+\bigO(1/t)),
 ~\text{as $t \to \infty$}.
\]

%%%%%%%%%%%%%%%%%%%%%%%%%%%%%%%%%%%%%%%%%%%%%%%%%% LEMMA
\begin{lemma}\label{lem:nullspace}
Provided $\alpha \neq 0$, $\beta \neq 0$,
$\Re(\overline{\alpha}\beta) \neq0$, the
source pairs are chosen with the above strategy and
assumption~\ref{assump:rankcond} holds, the matrix $\mM \equiv
\mM(\vec{\vx}_r,\omega)$ satisfies
\begin{equation}
\nullspace \mM = \linspan \left\{
\begin{bmatrix} -\Im\big(\vg_0(\vec{\vx}_r,\omega)\big)\\ 
\Re\big(\vg_0(\vec{\vx}_r,\omega)\big)\end{bmatrix}\right\}.
\end{equation}
\end{lemma}
%%%%%%%%%%%%%%%%%%%%%%%%%%%%%%%%%%%%%%%%%%%%%%%%%%

\begin{proof}
For clarity of exposition, we adopt the notation
\[
a_i = \Re\big(\vg_0\big)_i, ~b_i =
\Im\big(\vg_0\big)_i,
\]
for $i=1,\ldots,N$ and with $\vg_0 \equiv \vg_0(\vec{\vx}_r,\omega)$.
The proposed vector spanning the nullspace is $[\vv^\tr,\vw^\tr]^\tr = [ -
\Im(\vg_0)^\tr, \Re(\vg_0)^\tr]^\tr$ and has components $v_i = -b_i$ and
$w_i = a_i$ for $i=1,\ldots,N$. 

The proof is by induction on the number of sources $N$. For the purpose
of the induction argument, we denote by $\mM^{(N)}$ the measurement
matrix $\mM(\vec{\vx_r},\omega)$ corresponding to $N$ sources, which if
we use the strategy explained above, must be a $2N \times 2N$ real
matrix. For the base case $N=5$ of the induction, $\mM^{(5)}$ can be
written as 
\[
\mM^{(5)} = 
\begin{bmatrix}
A_2^- & A_1^+ & 0 & 0 & 0 & 
B_2^+ & B_1^- & 0 & 0 & 0 \\
A_3^- & 0 & A_1^+ & 0 & 0  &  
B_3^+ & 0 & B_1^- & 0 & 0  \\
A_4^- & 0 & 0 & A_1^+  & 0 & 
B_4^+ & 0 & 0 & B_1^-  & 0 \\
A_5^- & 0 & 0 & 0 & A_1^+  & 
B_5^+ & 0 & 0 & 0 & B_1^-  \\
0 & A_3^- & A_2^+ & 0 & 0 & 
0 & B_3^+ & B_2^- & 0 & 0 \\
0 & A_4^- & 0 & A_2^+ & 0 & 
0 & B_4^+ & 0 & B_2^- & 0 \\
0 & A_5^- & 0 & 0 & A_2^+ & 
0 & B_5^+ & 0 & 0 & B_2^- \\
0 & 0 & A_4^- & A_3^+ & 0 & 
0 & 0 & B_4^+ & B_3^- & 0 \\
0 & 0 & A_5^- & 0 & A_3^+ & 
0 & 0 & B_5^+ & 0 & B_3^- \\
0 & 0 & 0 & A_5^- & A_4^+ & 
0 & 0 & 0 & B_5^+ & B_4^- \\
\end{bmatrix},
\]
where we have used
\begin{equation}\label{eq:realimagM}
\begin{aligned}
A_i^{\pm} &= |\alpha||\beta| (\cos(\phi) a_i\pm\sin(\phi) b_i),\qquad 
B_i^{\pm} &= |\alpha||\beta| (\cos(\phi) b_i\pm\sin(\phi) a_i).
\end{aligned}
\end{equation}
Using the expressions \eqref{eq:realimagM}, the leading principal
$9\times 9$ minor of $\mM^{(5)}$ is 
\[
\begin{aligned}
|\mM_{1:9,1:9}^{(5)}| &=
-4|\alpha|^9|\beta|^9 \cos^2(\phi)\left(\cos(\phi)(b_3a_1-b_1a_3)+\sin(\phi)(b_3b_1+a_3a_1)\right)\times\\
&\qquad a_5(b_3a_2-b_2a_3)(b_2a_1-a_2b_1)(b_5a_4-a_5b_4).
\end{aligned}
\]
Therefore if assumption~\ref{assump:rankcond} holds and $\cos\phi\neq 0$
(which we get from $\Re(\overline{\alpha} \beta) \neq 0$),
we must have $\rank \mM^{(5)} \geq 9$. By direct calculations, we have
that
\[
\nullspace \mM^{(5)} = \linspan \left\{ [
-b_1, \ldots, -b_5,
a_1, \ldots, a_5 ]^\tr \right\}.
\]
Thus the base case $N=5$ holds and $\rank \mM^{(5)} = 9$.

For the induction hypothesis we assume that $N\geq5$ and that 
\[
\nullspace \mM^{(N)} = \linspan \left\{ [-\vb^\tr, \va^\tr]^\tr
\right\},
\]
where $\va = [a_1,\ldots,a_N]^\tr$ and $\vb = [b_1,\ldots,b_N]^\tr$. If
the first $2N$ source pairs to construct $\mM^{(N+1)}$ are chosen in
exactly the same way as the source pairs to construct $\mM^{(N)}$, and
the last two source pairs are, e.g. $(N+1,1)$ and $(N+1,2)$ we must
have for any $\vv,\vw \in \real^N$ and $v_{N+1},w_{N+1} \in \real$
that
\begin{equation}
 \mM^{(N+1)} \begin{bmatrix} \vv\\v_{N+1}\\\vw\\w_{N+1}\end{bmatrix}
 = \begin{bmatrix} \mM^{(N)} \begin{bmatrix}\vv\\\vw\end{bmatrix} \\
  A_{N+1}^-v_1 + A_1^+v_{N+1} + B_{N+1}^+w_1 + B_1^-w_{N+1}\\
  A_{N+1}^-v_2 + A_2^+v_{N+1} + B_{N+1}^+w_2 + B_2^-w_{N+1}
  \end{bmatrix}.
  \label{eq:induction}
\end{equation}
Hence if $[\vv^\tr,v_{N+1},\vw^\tr,w_{N+1}]^\tr\in \nullspace \mM^{(N+1)}$,
then we must have $[\vv^\tr,\vw^\tr]^\tr \in \nullspace
\mM^{(N)}$, i.e. there is some real $k\neq 0$ such that $\vv = -k \vb$
and $\vw = k \va$. Equating the last two components of
\eqref{eq:induction} to zero and using that $v_i = -kb_i$ and
$w_i = ka_i$ for $i=1,2$, one gets the linear system
\[
\begin{bmatrix}A_1^+ & B_1^-\\ A_2^+ & B_2^-\end{bmatrix}\begin{bmatrix}v_{N+1}\\ w_{N+1}\end{bmatrix} =
\begin{bmatrix}
kA_{N+1}^-b_1-kB_{N+1}^+a_1\\
kA_{N+1}^-b_2-kB_{N+1}^+a_2
\end{bmatrix}.
\]
Since
$A_1^+B_2^--A_2^+B_1^-=|\alpha|^2|\beta|^2(a_1b_2-a_2b_1)\neq0$,
the unique solution to this system is 
$v_{N+1} = -kb_{N+1}$ and $w_{N+1} = ka_{N+1}$. 
Thus the desired result holds for any $N\geq 5$.
\end{proof}

In figure~\ref{fig:mcond}, we show the condition number of
$\mM(\vec{\vx}_r,\omega)$ (i.e. $\sigma_1/\sigma_{2N-1}$ the ratio of
the largest singular value to the smallest non-zero singular value )
over a frequency band. The experimental setup is that given in
\S\ref{sec:numerics} and corresponds to sending exactly the same signal
from both locations in a source pair (i.e. $\alpha = \beta$ and $\phi =
0$). Figure~\ref{fig:mcond}(a) shows the condition number
of $\mM$ with $\vec{\vx}_r$ chosen so that
assumption~\ref{assump:rankcond} is satisfied, while
figure~\ref{fig:mcond}(b) shows the condition number of $\mM$ with
$\vec{\vx}_r$ chosen so that assumption~\ref{assump:rankcond} is
violated for some frequencies. In both cases, we see improved
conditioning by using more than $2N$ source pair experiments.

\begin{figure}
\centering
\begin{tabular}{cc}
(a) & (b)\\
\includegraphics[width=0.5\textwidth]{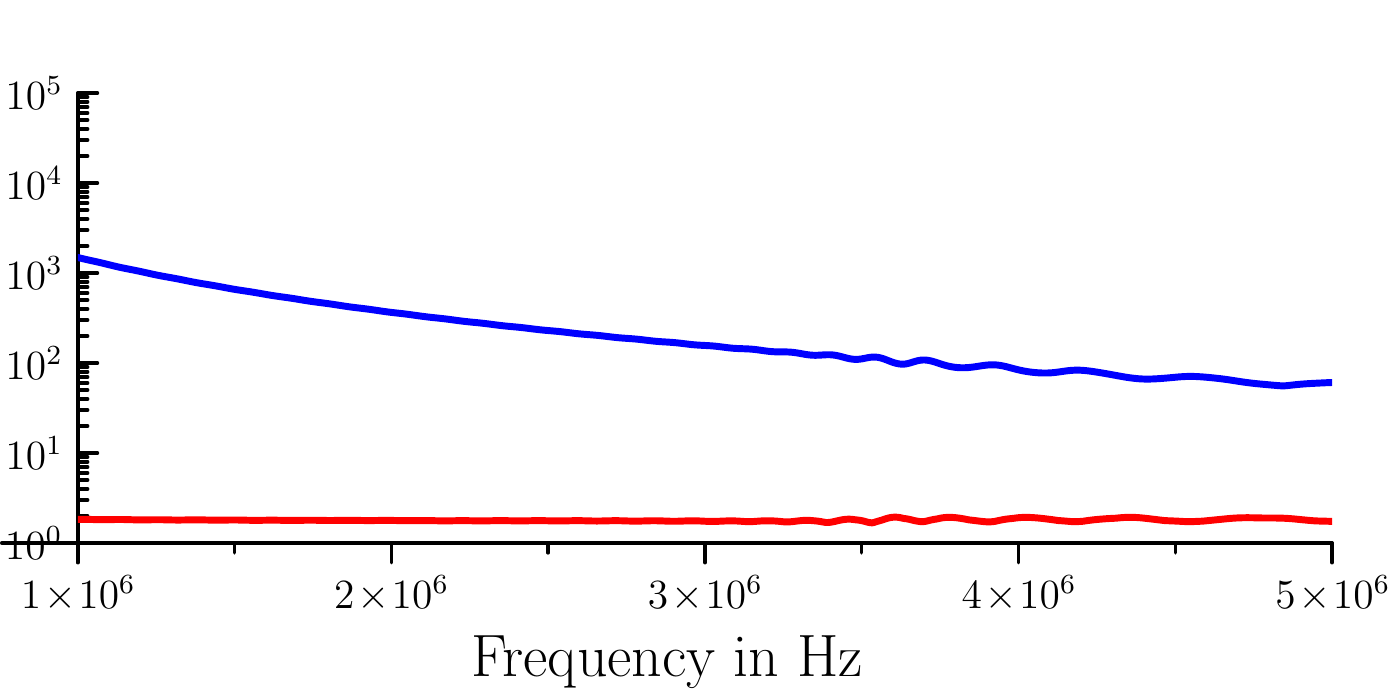} &
\includegraphics[width=0.5\textwidth]{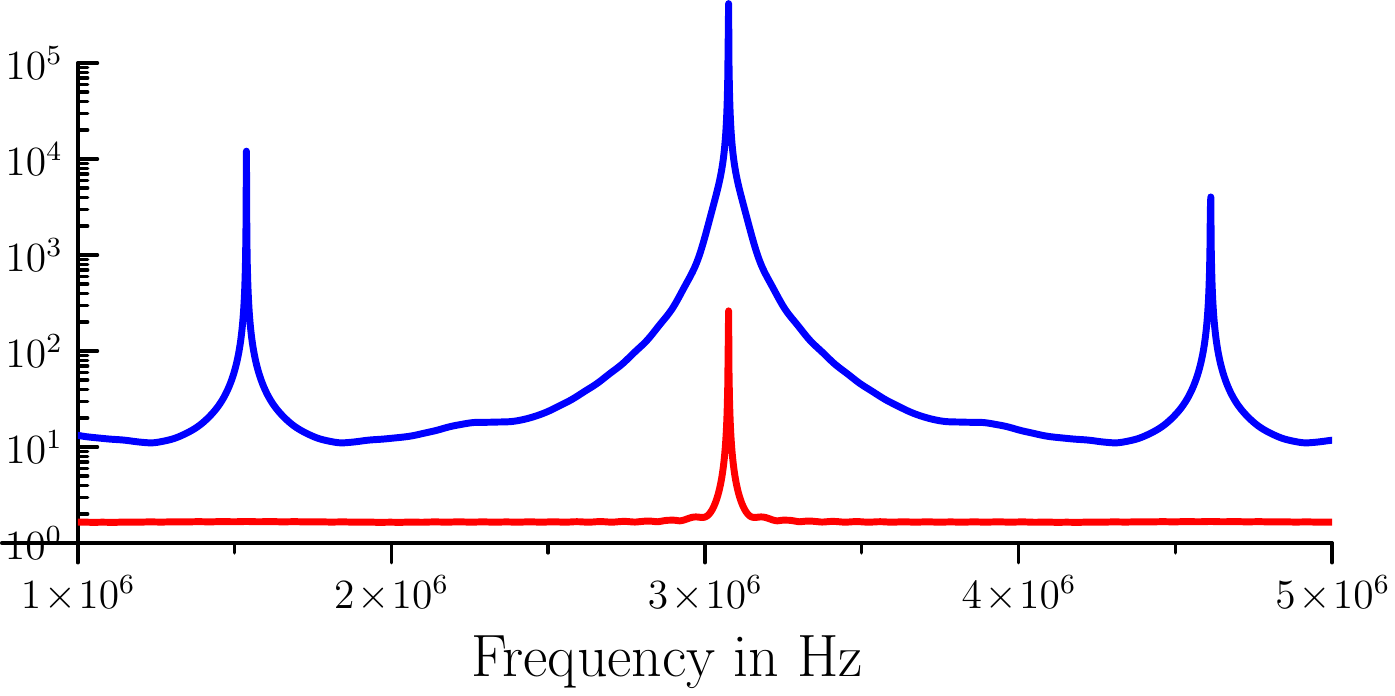}
\end{tabular}
\caption{Condition number of $\mM(\vec{\vx}_r,\omega)$ with receiver
location $\vec{\vx}_r$ chosen so that (a)
assumption~\ref{assump:rankcond} is satisfied, (b)
assumption~\ref{assump:rankcond} is violated for some frequencies. The
number of source pair experiments used is  $N_p=N(N-1)/2$ (in red) and
 $N_p=2N$ (in blue).} \label{fig:mcond}
\end{figure}

We now tie $\mM^\dagger \vd$ to the array response vector $\vp$.

%%%%%%%%%%%%%%%%%%%%%%%%%%%%%%%%%%%%%%%%%%%%%%%%%% THEOREM
\begin{theorem}\label{thm:invertibility}
Under the assumptions of lemma~\ref{lem:nullspace} it is possible to
recover $\vp \equiv \vp(\vec{\vx}_r,\omega)$ from the intensity data
$\vd$ up to a complex scalar multiple of $\vg_0 \equiv
\vg_0(\vec{\vx}_r,\omega)$, more precisely, $\mM^\dagger \vd$
determines the vector $\vp+\zeta\vg_0$ where
\begin{equation}\label{eq:scalarproj}
 \zeta \equiv \zeta(\vec{\vx}_r,\omega) = \frac{1}{2} -
\ii\frac{\Im(\vg_0^*\vp)}{\vg_0^*\vg_0}.
\end{equation}
\end{theorem}
%%%%%%%%%%%%%%%%%%%%%%%%%%%%%%%%%%%%%%%%%%%%%%%%%%
\begin{proof}
Recalling the form of our data we have
\[
\vd = 
\mM
\begin{bmatrix}
\Re\big(\vg_0+2\vp\big)
\\
\Im\big(\vg_0+2\vp\big)\end{bmatrix}.
\]
By lemma~\ref{lem:nullspace}, the matrix $\mM$ has a one dimensional
nullspace therefore 
\[
\mM^\dagger \vd =
\begin{bmatrix}
\Re(\vg_0+2\vp)
\\
\Im(\vg_0+2\vp)\end{bmatrix}-\tilde{\zeta}\begin{bmatrix}-\Im(\vg_0)\\\Re(\vg_0)\end{bmatrix},
\]
where  $\tilde{\zeta}\in\real$ is found by enforcing orthogonality with
$[-\Im(\vg_0)^\tr, \Re(\vg_0)^\tr]^\tr$, i.e.
\[
\tilde{\zeta} = \frac{1}{\vg_0^*\vg_0}[-\Re(\vg_0 + 2\vp)^\tr \Im(\vg_0) + \Im(\vg_0 +
2\vp)^\tr \Re(\vg_0)] = 
\frac{2\Im(\vg_0^*\vp)}{\vg_0^*\vg_0}.
\]
Thus from $\mM^\dagger \vd$ we can get the
$\complex^N$ vector
\[
\frac{1}{2}[\Re(\vg_0+2\vp)+\tilde{\zeta}\Im(\vg_0)]
+\frac{\ii}{2}[\Im(\vg_0+2\vp)-\tilde{\zeta}\Re(\vg_0)]
=\frac{1}{2}\vg_0+\vp-\frac{\ii}{2}\tilde{\zeta}\vg_0
= \vp+\zeta\vg_0,
\]
where the scalar $\zeta \equiv \zeta(\vec{\vx}_r,\omega)\in\complex$ is
given by \eqref{eq:scalarproj}.
\end{proof}

%%%%%%%%%%%%%%%%%%%%%%%%%%%%%%%%%%%%%%%%%%%%%%%%%%%%%%%%%%%%%%%%%%%%%%%%
\section{Kirchhoff migration imaging}\label{sec:migration} 
We now show that we can image with the reconstructed field
$\vp+\zeta\vg_0$ instead of $\vp$ by using Kirchhoff migration. This
is because the Kirchhoff migration image of $\zeta\vg_0$ is negligible
compared to the image of $\vp$ for high frequencies. In order to show
that this nullspace vector does not affect the imaging, we need to make
sure the receiver satisfies the following condition.

%%%%%%%%%%%%%%%%%%%%%%%%%%%%%%%%%%%%%%%%%%%%%%%%%% ASSUMPTION
\begin{assumption}[Geometric imaging
conditions]\label{assump:geometricimg} For a scattering potential
with support contained inside an image window $\sW$, we assume $\vec{\vx}_r$ satisfies 
\begin{equation}\label{eq:geometricimg}
\frac{\vec{\vx}_s-\vec{\vx}_r}{|\vec{\vx}_s-\vec{\vx}_r|} \neq 
\frac{\vec{\vx}_s-\vec{\vy}}{|\vec{\vx}_s-\vec{\vy}|},
\end{equation}
for $s=1,\ldots,N$ and $\vec{\vy}\in \sW$. \end{assumption}
%%%%%%%%%%%%%%%%%%%%%%%%%%%%%%%%%%%%%%%%%%%%%%%%%%

One way to guarantee assumption~\ref{assump:geometricimg} holds is to
place the receiver at location $\vec{\vx}_r$ outside of the shaded
region in figure~\ref{fig:geoimg}.

\begin{figure}
\centering
\includegraphics[scale=0.8]{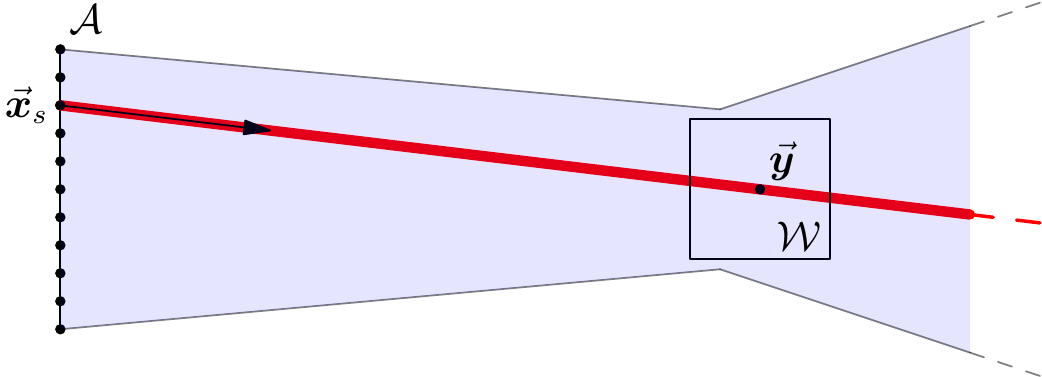}
\caption{Given an array $\sA$ and a region $\sW$ containing the
scatterers to image, assumption~\ref{assump:geometricimg} ensures the
receiver location $\vec{\vx}_r$ is outside of the
blue shaded region. This guarantees the Kirchhoff images using data $\vp$
and the recovered $\vp + \zeta \vg_0$ are essentially the same. The positive ray in the
direction $\vec{\vx}_s  - \vec{\vy}$ for particular $\vec{\vx}_s\in\sA$
and $\vec{\vy}\in \sW$ is indicated in red. If $\vec{\vx}_r$ is outside
the blue shaded region, we have
$(\vec{\vx}_s-\vec{\vx}_r)/|\vec{\vx}_s-\vec{\vx}_r| \neq
(\vec{\vx}_s-\vec{\vy})/|\vec{\vx}_s-\vec{\vy}|$ for all
$\vec{\vx}_s\in\sA$ and all $\vec{\vy}\in \sW$. }\label{fig:geoimg}
\end{figure}

%%%%%%%%%%%%%%%%%%%%%%%%%%%%%%%%%%%%%%%%%%%%%%%%%% THEOREM
\begin{theorem}\label{thm:migration}
Provided assumption~\ref{assump:geometricimg} holds, the image of the reconstructed array response vector is
\[
\KM[\vp+\zeta\vg_0,\omega](\vec{\vy}) \approx \KM[\vp,\omega](\vec{\vy}).
\]
\end{theorem}
%%%%%%%%%%%%%%%%%%%%%%%%%%%%%%%%%%%%%%%%%%%%%%%%%%
\begin{proof}
First we approximate the Kirchhoff imaging functional
\eqref{eq:kirchoff} by an integral over the array $\sA$, i.e.
\begin{equation}
\begin{aligned}
\KM[\zeta\vg_0,\omega](\vec{\vy}) &=
\overline{\hat{G}}(\vec{\vx}_r,\vec{\vy},\omega)\vg_0(\vec{\vy},\omega)^*\zeta(\vec{\vx}_r,\omega)\vg_0(\vec{\vx}_r,\omega)\\
&\sim
\zeta(\vec{\vx}_r,\omega)\int_{\sA}d\vx_s C(\vx_s) \exp\left(\ii\omega c_0^{-1} \big(|\vec{\vx}_s-\vec{\vx}_r|-|\vec{\vx}_s-\vec{\vy}|-|\vec{\vy}-\vec{\vx}_r|\big)\right),
\end{aligned}
\label{eq:kmg0}
\end{equation}
where the symbol $\sim$ means equal up to a constant and $C(\vx_s)$
collects smooth geometric spreading terms.

Let us first use the stationary phase method (see e.g.
\cite{Bleistein:2001:MMI}) on the integral over $\sA$. In the high
frequency limit $\omega \to \infty$, the dominant contribution comes
from stationary points of the phase, i.e. the points
$\vec{\vx}_s$ for which  
\[\nabla_{\vec{\vx}_s}\Big(
|\vec{\vx}_s-\vec{\vx}_r|-|\vec{\vx}_s-\vec{\vy}|-|\vec{\vy}-\vec{\vx}_r|\Big)
= 0.
\]
The stationary points must then satisfy
\[
\frac{\vec{\vx}_s-\vec{\vx}_r}{|\vec{\vx}_s-\vec{\vx}_r|} =
\frac{\vec{\vx}_s-\vec{\vy}}{|\vec{\vx}_s-\vec{\vy}|}.
\]
Thus by assumption~\ref{assump:geometricimg}, there are no stationary
points in the phase of the integral over the array $\sA$ appearing in
\eqref{eq:kmg0}. Neglecting boundary effects, this integral goes to zero
faster than any polynomial in $\omega$ (see e.g.
\cite{Bleistein:1986:AEI}).

We now show that in the high frequency limit $\omega \to \infty$,  we have
$\zeta(\vec{\vx}_r,\omega) \to 1/2$. 
Recalling \eqref{eq:scalarproj}, we have
\begin{equation}\label{eq:zetalim}
\begin{aligned}
\zeta(\vec{\vx}_r,\omega) &= \frac{1}{2} +
\frac{\vg_0(\vec{\vx}_r,\omega)^*\vp(\vec{\vx}_r,\omega) -
\vp(\vec{\vx}_r,\omega)^*\vg_0(\vec{\vx}_r,\omega)}{\vg_0(\vec{\vx}_r,\omega)^*\vg_0(\vec{\vx}_r,\omega)}\\
&\sim \frac{1}{2} + \frac{\omega^2}{c_0^2}\int d\vec{\vz}\int_{\sA}d\vx_s
 C(\vx_s)
\exp\left(\ii\omega c_0^{-1} \big(|\vec{\vx}_s-\vec{\vz}|+|\vec{\vz}-\vec{\vx}_r|-|\vec{\vx}_s-\vec{\vx}_r|\big)\right)\\
&\qquad - \frac{\omega^2}{c_0^2}\int d\vec{\vz}\int_{\sA}d\vx_s  C(\vx_s)
\exp\left(\ii\omega c_0^{-1}
\big(|\vec{\vx}_s-\vec{\vx}_r|-|\vec{\vx}_s-\vec{\vz}|-|\vec{\vz}-\vec{\vx}_r|\big)\right),\\
\end{aligned}
\end{equation}
where $C(\vec{\vx}_s)$ collects geometric spreading terms and
$|\vg_0(\vec{\vx}_r,\omega)|^{-2}$, which is actually independent of the
frequency $\omega$.  By assumption~\ref{assump:geometricimg}, the
integrals over $\sA$ in \eqref{eq:zetalim} do not have any stationary points.  Thus
if we neglect boundary terms, these integrals must go to zero faster
than any polynomial in $\omega$ (see e.g. \cite{Bleistein:1986:AEI}), meaning that
$\zeta(\vec{\vx}_r,\omega) \to 1/2$ as $\omega \to \infty$. Thus
$\KM[\zeta\vg_0,\omega](\vec{\vy}) \to 0$ as $\omega \to \infty$.
\end{proof}

%%%%%%%%%%%%%%%%%%%%%%%%%%%%%%%%%%%%%%%%%%%%%%%%%%%%%%%%%%%%%%%%%%%%%%%%
\section{Autocorrelation measurements}\label{sec:stochillum} 
Up to this point we have assumed deterministic control over the source
illuminations. In this section we relax this control by driving the array with
stochastic signals. We start in section~\ref{sec:stoch:ergo} by recalling an
ergodicity result of Garnier and Papanicolaou~\cite{Garnier:2009:PSI} which
guarantees that if Gaussian stochastic processes are used to drive the sources,
the realization average of the total field can be well approximated by time
averages of the total field. Then in section~\ref{sec:stoch:pair} we adapt the
source pair illumination strategy to pairs of sources driven by two correlated
Gaussian processes, with (known) correlation identical for different pairs.
From these pairwise illuminations we measure empirical autocorrelations to
obtain intensity measurements that are essentially (up to ergodic averaging)
the same as those using the deterministic strategy of
section~\ref{sec:detillum}. 

%%%%%%%%%%%%%%%%%%%%%%%%%%%%%%%%%%%%%%%%%%%%%%%%%%%%%%%%%%%%%%%%%%%%%%%%
\subsection{Stochastic array illuminations} \label{sec:stoch:ergo}
We consider array illuminations $\vf(t) \in \complex^N$ given by a
stationary Gaussian process with mean zero and with correlation the $N
\times N$ matrix function 
\begin{equation}\label{eq:noise}
\mR(\tau) = \langle \overline{\vf}(t)\vf^\tr(t+\tau)\rangle.
\end{equation}
Here $\langle \cdot \rangle$ denotes the expectation with respect to
realizations of $\vf$, and in an abuse of notation we have denoted by
$\vf(t)$ the time domain vector of signals driving the array. Since
$R_{s,s'}(\tau) = \langle
\overline{f}_s(t)f_{s'}(t+\tau)\rangle=\overline{\langle\overline{f}_{s'}(t+\tau)f_s(t)\rangle}=\overline{R_{s',s}(-\tau)}$
for $s,s'=1,\ldots,N$, we have $\mR(\tau) = \mR^*(-\tau)$ and so
$\hat{\mR}(\omega)$ is a Hermitian $N\times N$ matrix.

The total field $u$ at the receiver arising from the array illumination $\vf$ is, in the time domain,
\begin{equation}\label{eq:totalfieldstoch}
u(\vec{\vx}_r,t) = \sum_{s=1}^N \int dt'
G(\vec{\vx}_r,\vec{\vx}_s,t-t')f_s(t'),
\end{equation}
where $G$ is the Born approximation of the inhomogeneous Green function,
i.e.
\[
G(\vec{\vx}_r,\vec{\vx}_s,t) = \frac{1}{2\pi}\int d\omega e^{-\ii\omega
t}\Big[
\hat{G}_0(\vec{\vx}_r,\vec{\vx}_s,\omega)+k^2\int
d\vec{\vz}\rho(\vec{\vz})\hat{G}_0(\vec{\vx}_r,\vec{\vz},\omega)\hat{G}_0(\vec{\vz},\vec{\vx}_s,\omega)\Big].
\]
The empirical autocorrelation of $u$ is
\begin{equation}\label{eq:empautocorr}
\C(\vec{\vx}_r,\tau) = \frac{1}{2T}\int_{-T}^T \overline{u}(\vec{\vx}_r,t)u(\vec{\vx}_r,t+\tau)dt,
\end{equation}
where $T$ is a known measurement time. Following Garnier and
Papanicolaou~\cite{Garnier:2009:PSI}, we formulate
proposition~\ref{prop:statstab} regarding the statistical stability and
ergodicity of \eqref{eq:empautocorr}. This proposition is essentially
the same as \cite[Proposition 4.1]{Garnier:2009:PSI}, but we make small
modifications to allow for complex fields and more general correlations
in space.  We include it here for the sake of completeness and the proof
can be found in appendix~\ref{sec:appendix}.

%%%%%%%%%%%%%%%%%%%%%%%%%%%%%%%%%%%%%%%%%%%%%%%%%% THEOREM
\begin{proposition}\label{prop:statstab}
Assume $\vf$ satisfies \eqref{eq:noise}. The expectation (w.r.t.
realizations of $\vf$) of the empirical autocorrelation
\eqref{eq:empautocorr} is independent of measurement time $T$:
\begin{equation}\label{eq:indT}
\langle \C(\vec{\vx}_r,\tau)\rangle = \CE(\vec{\vx}_r,\tau),
\end{equation}
where
\begin{equation}\label{eq:expectedcorr}
\begin{aligned}
\CE(\vec{\vx}_r,\tau) &= \sum_{s,s'=1}^N
\int dt' \int dt''\overline{G}(\vec{\vx}_r,\vec{\vx}_r,-t')G(\vec{\vx}_r,\vec{\vx}_{s'},\tau-t'')R_{s,s'}(t''-t')\\
&=\frac{1}{2\pi}\int d\omega
e^{-\ii\omega\tau}\vg(\vec{\vx}_r,\omega)^*\hat{\mR}(\omega)\vg(\vec{\vx}_r,\omega).\\
\end{aligned}
\end{equation}
Furthermore, \eqref{eq:empautocorr} is ergodic, i.e. 
\begin{equation}\label{eq:ergodicity}
\C(\vec{\vx}_r,\tau)\xrightarrow{T\to\infty}\CE(\vec{\vx}_r,\tau).
\end{equation}
\end{proposition}
%%%%%%%%%%%%%%%%%%%%%%%%%%%%%%%%%%%%%%%%%%%%%%%%%%
\subsection{Pairwise stochastic illuminations} \label{sec:stoch:pair}

We make $N_p$ illuminations each corresponding to using only two distinct sources $(i(m),j(m)) \in \{1,\ldots,N\}^2$, $m=1,\ldots,N_p$. The correlation matrix for the $m-$th experiment has the form
\begin{equation}
 \hat{\mR}_m(\omega) = \mF_m \mC(\omega) \mF_m^\tr,
 \label{eq:spatialcorrfn}
\end{equation}
where $\mF_m = [ \ve_{i(m)}, \ve_{j(m)} ] \in \real^{N \times 2}$ and
$\mC(\omega)$ is a known $2\times 2$ Hermitian positive semidefinite matrix that
represents the correlation between the two sources and is assumed to be the same for all
experiments. For instance, if we send the same signal with power spectrum $F(\omega)$ from both sources in a pair, this correlation matrix is
\[
 \mC(\omega) = F(\omega) \begin{bmatrix} 1 & 1 \\ 1 & 1 \end{bmatrix}.
\]

By the ergodicity \eqref{eq:ergodicity} of proposition~\ref{prop:statstab},
when we measure the empirical autocorrelation $\C_m$ of $u_m$ at the receiver $\vec{\vx}_r$ for long enough time $T$, the empirical autocorrelation is close to an
intensity measurement, i.e.
\begin{equation}\label{eq:stoc_pair_meas}
\hat{\CE}_m(\vec{\vx}_r,\omega) = \vg(\vec{\vx}_r,\omega)^* \mF_m \mC(\omega) \mF_m^\tr \vg(\vec{\vx}_r,\omega).
\end{equation}
By using appropriate single source illuminations driven by a signal with known
correlation, it is possible to measure
\begin{equation}\label{eq:stoc_sing_meas}
\hat{\CE}_i^0(\vec{\vx}_r,\omega) = 
\vg^*(\vec{\vx}_r,\omega)\ve_{i}\ve_i^\tr\vg(\vec{\vx}_r,\omega).
~\text{ for $i=1,\ldots,N$.}
\end{equation}
From \eqref{eq:stoc_pair_meas} and \eqref{eq:stoc_sing_meas} we obtain the
$m-$th measurement 
\begin{equation}\label{eq:mthmeas}
\begin{aligned}
d_m(\vec{\vx}_r,\omega) &=  \hat{\CE}_m(\vec{\vx}_r,\omega) -
 C_{11}(\omega) \hat{\CE}_{i(m)}^0(\vec{\vx}_r,\omega) -
 C_{22}(\omega) \hat{\CE}_{j(m)}^0(\vec{\vx}_r,\omega)
\\
&=\vg(\vec{\vx}_r,\omega)^*\mF_m \mD(\omega) \mF_m^\tr \vg(\vec{\vx}_r,\omega),
\end{aligned}
\end{equation}
where the matrix $\mD$ is $2\times 2$, Hermitian with zero diagonal, i.e.
precisely of the same form as the matrix $\mD$ we encountered in the intensity
measurements case \eqref{eq:mD}.

Proceeding analogously as in section~\ref{sec:detillum} and recalling that
$\vg=\vg_0+\vp$ we have 
\[
\begin{aligned}
d_m(\vec{\vx}_r,\omega) &=
\big(\vg_0+\vp\big)^*\mF_m \mD(\omega) \mF_m^\tr(\vg_0+\vp\big).
\end{aligned}
\]
Collecting the measurements for $m=1,\ldots,N_p$ and
neglecting the quadratic term in $\vp$ we have the approximate data
\begin{equation}\label{eq:stocdata}
\begin{aligned}
\begin{bmatrix} 
d_1(\vec{\vx}_r,\omega)\\
d_2(\vec{\vx}_r,\omega)\\
\vdots\\
d_{N_p}(\vec{\vx}_r,\omega)
\end{bmatrix} \approx
\vd(\vec{\vx}_r,\omega)=
\mM(\vec{\vx}_r,\omega)
\begin{bmatrix}
\Re\big(\vg_0+2\vp\big)\\
\Im\big(\vg_0+2\vp\big)
\end{bmatrix},
\end{aligned}
\end{equation}
where the matrix $\mM\in\real^{N_p\times 2N}$ is again given by
\eqref{eq:M}. Thus, the data \eqref{eq:stocdata} obtained by measuring
the empirical autocorrelation \eqref{eq:empautocorr} and using
correlated pair illuminations, is essentially the same as the data
obtained using deterministic source pairs \eqref{eq:detdata}. Hence the
analysis of the matrix $\mM$ of \S\ref{sec:invertible} holds and we can
use Kirchhoff migration as we did in \S\ref{sec:migration} for the
intensity measurements case.

%%%%%%%%%%%%%%%%%%%%%%%%%%%%%%%%%%%%%%%%%%%%%%%%%% REMARK (array noise)
\begin{remark}[Uncorrelated background
illumination]\label{rem:backgroundnoise} The proposed illumination
strategy is robust with respect to noise and even allows to send the \emph{same}
Gaussian signal from the $m-$th source pair 
$(i(m),j(m))$ and \emph{independent} Gaussian
signals from all remaining sources on the array. If the independent
signals have the same spectral density $F(\omega)$ as the source pair
signal, the correlation matrix for the $m-$th experiment is
\begin{equation}\label{eq:inhomog_meas}
\hat{\mR}_m(\omega) = 
F(\omega)\left(\mI + \mF_m \begin{bmatrix}0&1\\1&0\end{bmatrix} \mF_m^\tr\right),
\end{equation}
where $\mI$ is $N\times N$ identity matrix. By subtracting from the
autocorrelation for the $m-$th experiment, the autocorrelation for a
reference illumination that sends \emph{independent} Gaussian signals with
correlation matrix $F(\omega) \mI$, it is possible to obtain $m-$th measurement
\eqref{eq:mthmeas} with \[\mD(\omega) = F(\omega)\begin{bmatrix} 0&1\\1&0\end{bmatrix}.\]
\end{remark}
%%%%%%%%%%%%%%%%%%%%%%%%%%%%%%%%%%%%%%%%%%%%%%%%%%

%%%%%%%%%%%%%%%%%%%%%%%%%%%%%%%%%%%%%%%%%%%%%%%%%%%%%%%%%%%%%%%%%%%%%%%%%
%%% SNR
\section{Additive noise}\label{sec:snr}
Here we discuss the effects of additive instrumental noise in autocorrelated
measurements of the total field. The
total field at $\vec{\vx}_r$ resulting from illuminating with the $m-$th pair
and tainted with additive noise is $u_m(\vec{\vx}_r,t)+\xi(t)$. We assume the
noise $\xi$ is a stationary Gaussian process with mean zero and spectral
density 
\begin{equation}\label{eq:noisepsd}
\hat{\Xi}(\omega) =
\exp\left(\frac{-l_c^2(\omega-\omega_0)^2}{4\pi}\right).
\end{equation}
Here $l_c$ represents the correlation time of the noise (i.e. $\Xi(\tau)=\langle
\overline{\xi}(t)\xi(t+\tau)\rangle \approx 0$ for $\tau\gg l_c$) and $\omega_0$ is
the central angular frequency of the noise. If the noise $\xi$ is independent
of the signals used to drive the source pairs, it can be shown using the
techniques of appendix~\ref{sec:appendix} that 
\[
\frac{1}{2T}\int_{-T}^T dt
\big(\overline{u}_m(\vec{\vx}_r,t)+\overline{\xi}(t)\big)\big(u_m(\vec{\vx}_r,t+\tau)+\xi(t+\tau)\big)\xrightarrow{T\to\infty}
\CE_m(\vec{\vx}_r,\tau)+\Xi(\tau),
\]
where $\CE_m$ is given by \eqref{eq:expectedcorr}.

Assuming the same form of instrumental noise in the single source
reference measurements, the $m-$th measurement
$d_m(\vec{\vx}_r,\omega)$ is
\[
d_m(\vec{\vx}_r,\omega) =
\big(\vg_0+\vp\big)^*\mF_m \mD(\omega) \mF_m^\tr \big(\vg_0+\vp\big)+C\hat{\Xi}(\omega).
\]
for some $C\in\real$. Neglecting the terms which are quadratic in $\vp$
and going back to the time domain we have 
\[
\begin{aligned}
d_m(\vec{\vx}_r,\tau) \approx \frac{1}{2\pi}\int d\omega
e^{-\ii\omega\tau}\Big[&\vg_0(\vec{\vx}_r,\omega)^* \mF_m \mD(\omega) \mF_m^\tr \vg_0(\vec{\vx}_r,\omega)\\
+&\vg_0(\vec{\vx}_r,\omega)^* \mF_m \mD(\omega) \mF_m^\tr \vp(\vec{\vx}_r,\omega)\\
+&\vp(\vec{\vx}_r,\omega)^* \mF_m \mD(\omega) \mF_m^\tr \vg_0(\vec{\vx}_r,\omega)\Big]+C\Xi(\tau),
\end{aligned}
\]
with the slight abuse of notation of using $d_m$ for both time and frequency domain quantities. The second
and third terms in the integrand are incident-scattered field
correlations and contain the available information about the scattering
potential $\rho(\vec{\vy})$. 

For simplicity, we now focus on the case where the source pair signals have
correlation matrix 
\[
 \mD(\omega) = F(\omega) \begin{bmatrix} 0 & e^{\ii \omega \phi} \\
 e^{-\ii\omega\phi} & 0 \end{bmatrix}.
\]
Such correlation corresponds to sending a signal from one of the sources in a
pair and a copy of the same signal delayed by $\phi$ from the other source.
For a point scatterer at $\vec{\vy}$, the incident-scattered terms have peaks at
delay times $\tau(\vec{\vy})$ corresponding to differences between travel times
of a reflected path and direct path, i.e. for the $m-$th experiment the
peaks occur at the four possible delays
\[
\tau(\vec{\vy}) =
\begin{cases}
\pm((|\vec{\vx}_{j(m)}-\vec{\vy}|+|\vec{\vy}-\vec{\vx}_r|-|\vec{\vx}_{i(m)}-\vec{\vx}_r|)/c_0
+ \phi ),\\
\pm((|\vec{\vx}_{i(m)}-\vec{\vy}|+|\vec{\vy}-\vec{\vx}_r|-|\vec{\vx}_{j(m)}-\vec{\vx}_r|)/c_0
- \phi ).
\end{cases}
\]

Consider then the minimal delay time $\tau_{\min}(\vec{\vy})$ given by
\begin{equation}\label{eq:mintau}
\tau_{\min}(\vec{\vy}) =
\min_{\vec{\vx}_s,\vec{\vx}_{s'}\in\sA}\left|\frac{|\vec{\vx}_s-\vec{\vy}|+|\vec{\vy}-\vec{\vx}_r|-|\vec{\vx}_{s'}-\vec{\vx}_r|}{c_0}
\pm \phi\right|,
\end{equation}
that is the minimal delay time we expect the incident-scattered
correlations to peak. If we assume the additive noise decorrelates much
faster than the first incident-scattered arrival from $\vec{\vy}$ (i.e. $l_c\ll
\tau_{\min}(\vec{\vy})$), then the information of the scatterer
$\rho(\vec{\vy})$ contained in $d_m(\vec{\vx}_r,\tau)$ is
essentially unchanged (up to ergodic averaging). Hence we can stably
image using the proposed method at $\vec{\vy}$ provided
$\tau_{\min}(\vec{\vy})\gg l_c$.

%%%%%%%%%%%%%%%%%%%%%%%%%%%%%%%%%%%%%%%%%%%%%%%%%%%%%%%%%%%%%%%%%%%%%%%%%
%%% Numerical experiments

\section{Numerical experiments}\label{sec:numerics}
Here we include 2D numerical experiments of our proposed imaging routine
for scalings corresponding to acoustics (\S\ref{sec:numacoustic}) and
optics (\S\ref{sec:numoptic}). We demonstrate the stochastic source pair
illumination strategy for the acoustic regime, i.e. we compute the
autocorrelations for time domain data. In the optic regime this is an
expensive calculation, so we use instead power spectra (i.e.
deterministic illuminations).

\subsection{Acoustic regime}\label{sec:numacoustic}
For imaging in an acoustic regime, our choice of physical parameters
corresponds to ultrasound in water. We choose the background wave
velocity to be $c_0=1500$ m/s. The central frequency for all signals
(sources and additive noise) is $3$ MHz, which gives a central
wavelength of $\lambda_0 = 0.5$ mm. We center a source array $\sA$ at
the origin consisting of $41$ sources at coordinates
$\vec{\vx}_s=(0,-10\lambda_0+(s-1)\lambda_0/2)$ for $s=1,\ldots,41$. A
single receiver is located at the coordinate
$\vec{\vx}_r=(-20\lambda_0,-20\lambda_0)$ (see
figure~\ref{fig:arraysetup}).

We generate a stationary Gaussian time signal $f(t)$ with mean zero and
correlation function 
\[
F(\tau) = \exp\left(-\pi\frac{\tau^2}{t_c^2}\right),
\]
using the Wiener-Khinchin theorem. The correlation time $t_c\approx
1.25$ $\mu$s which gives the signal an effective frequency band $[1,5]$
MHz. We generate time signals of length $2T$ for $T\approx 260$ $\mu$s
with $8001$ uniformly spaced samples. This sampling is enough to resolve
the frequencies in the angular frequency band $\sB$, while $T$ is
sufficient to observe ergodic averaging (see \S\ref{sec:stochillum}). By
placing the \emph{same} realization of this signal $\hat{f}(\omega)$ at 
the locations $\vec{\vx}_{i(m)}$ and $\vec{\vx}_{j(m)}$ we generate
the pair illumination $\vf_m(\omega) =
\hat{f}(\omega)(\ve_{i(m)}+\ve_{j(m)}).$ Similarly, by placing an
independent realization of $\hat{f}(\omega)$ at location $\vec{\vx}_i$
we generate the single source reference illumination $\vf_i^0(\omega) =
\hat{f}(\omega)\ve_i$.

For all experiments, synthetic data is generated in the frequency domain
using the Born approximation. We assume 3D wave propagation for
simplicity so that $G_0$ is given by \eqref{eq:green} for $d=3$. The
$m-$th measurement is obtained through the formula
\[
d_m(\vec{\vx}_r,\omega) =
\hat{\CE}_m(\vec{\vx}_r,\omega)-\hat{\CE}_{i(m)}^0(\vec{\vx}_r,\omega)-\hat{\CE}_{j(m)}^0(\vec{\vx}_r,\omega),
\]
where 
\[
\begin{aligned}
\hat{\CE}_m(\vec{\vx}_r,\omega) &=
\Big|\big(\vg_0(\vec{\vx}_r,\omega)+\vp(\vec{\vx}_r,\omega)\big)^\tr\vf_m(\omega)\Big|^2,\\
\hat{\CE}_i^0(\vec{\vx}_r,\omega) &=
\Big|\big(\vg_0(\vec{\vx}_r,\omega)+\vp(\vec{\vx}_r,\omega)\big)^\tr\vf_i^0(\omega)\Big|^2,
\end{aligned}
\]
with $\vg_0$ and $\vp$ defined by \eqref{eq:homoggreen} and
\eqref{eq:arv} respectively.

For these simulations we use the full set of pair illuminations, which
for $N=41$ source locations, generates a measurement matrix
$\mM(\vec{\vx}_r,\omega)\in\real^{820\times 82}$. We use the
Moore Penrose pseudoinverse $\mM^\dagger$ to recover $\vp+\zeta\vg_0$
for each $\omega\in\sB$. When the number of sources $N$ and thus the
dimension of $\mM$ is large (recall $\mM\in\real^{N(N-1)/2\times N}$), the
pseudoinverse could become computationally expensive. However, the
system is sparse as it contains only $4$ non-zero elements per row, so
linear least square solvers that exploit sparsity (e.g.
CGLS \cite{Hansen:1998:RDD}) may be more efficient than our approach.
Furthermore, as discussed in \S\ref{sec:invertible} we can reduce the
size of $\mM$ to $2N\times 2N$ while keeping the nullspace of $\mM$
one-dimensional by using an appropriate subset of source pairs. 

We form an image at $\vec{\vy}\in
\sW=\{(100\lambda_0+i\lambda_0/2.5,j\lambda_0/2.5),\text{ for }
i,j=-25,\ldots,25\}$ using the Kirchhoff migration functional
(\S\ref{sec:km}), summed over the bandwidth band $\sB$,
\[
\KM[\vp+\zeta\vg_0](\vec{\vy}) = \int_{\sB} d\omega
\KM[\vp+\zeta\vg_0,\omega](\vec{\vy}).
\]
For our first experiment, we place a single point reflector at the
location $\vec{\vy}=(100\lambda_0,0)$, with refractive index
perturbation $\rho(\vec{\vy}) = 1\times 10^{-8}$. The migrated image
(figure~\ref{fig:pointref}a) indeed exhibits the cross-range
(Rayleigh) resolution estimate $\lambda_0L/a \approx 5\lambda_0 $ and
range resolution estimate $c_0/|\sB| \approx 1\lambda_0$. Note that
there is a trade-off in the choice of the reflectivity: $\rho$ has to be
sufficiently small so that the quadratic terms in $\vp$
can be neglected in \eqref{eq:detdata}. However the smaller $\rho$ is,
the longer the acquisition time $T$ has to be in order to better observe
the reflected-incident correlations in the data.

In our second experiment (figure~\ref{fig:pointref}b), we consider two
oblique reflectors located at $\vec{\vy}_1 = (99\lambda_0,-2\lambda_0)$
and $\vec{\vy}_2=(103\lambda_0,4\lambda_0)$ each with $\rho(\vec{\vy}_i)
= 1\times 10^{-8}$. We include a reconstruction of an extended scatterer
(line segment) in figure~\ref{fig:extended_scatterer}. Here the line
segment is generated as a set of point reflectors each with
$\rho(\vec{\vy}_i)=1\times 10^{-9}$ uniformly spaced by $\lambda_0/8$.

\begin{figure}[!hbtp]
\centering
\begin{tabular}{p{0.5\textwidth}p{0.5\textwidth}}
\centering $\KM[\vp](\vec{\vy})$ &
\centering $\KM[\vp+\zeta\vg_0](\vec{\vy})$
\end{tabular}
\includegraphics[width=\textwidth]{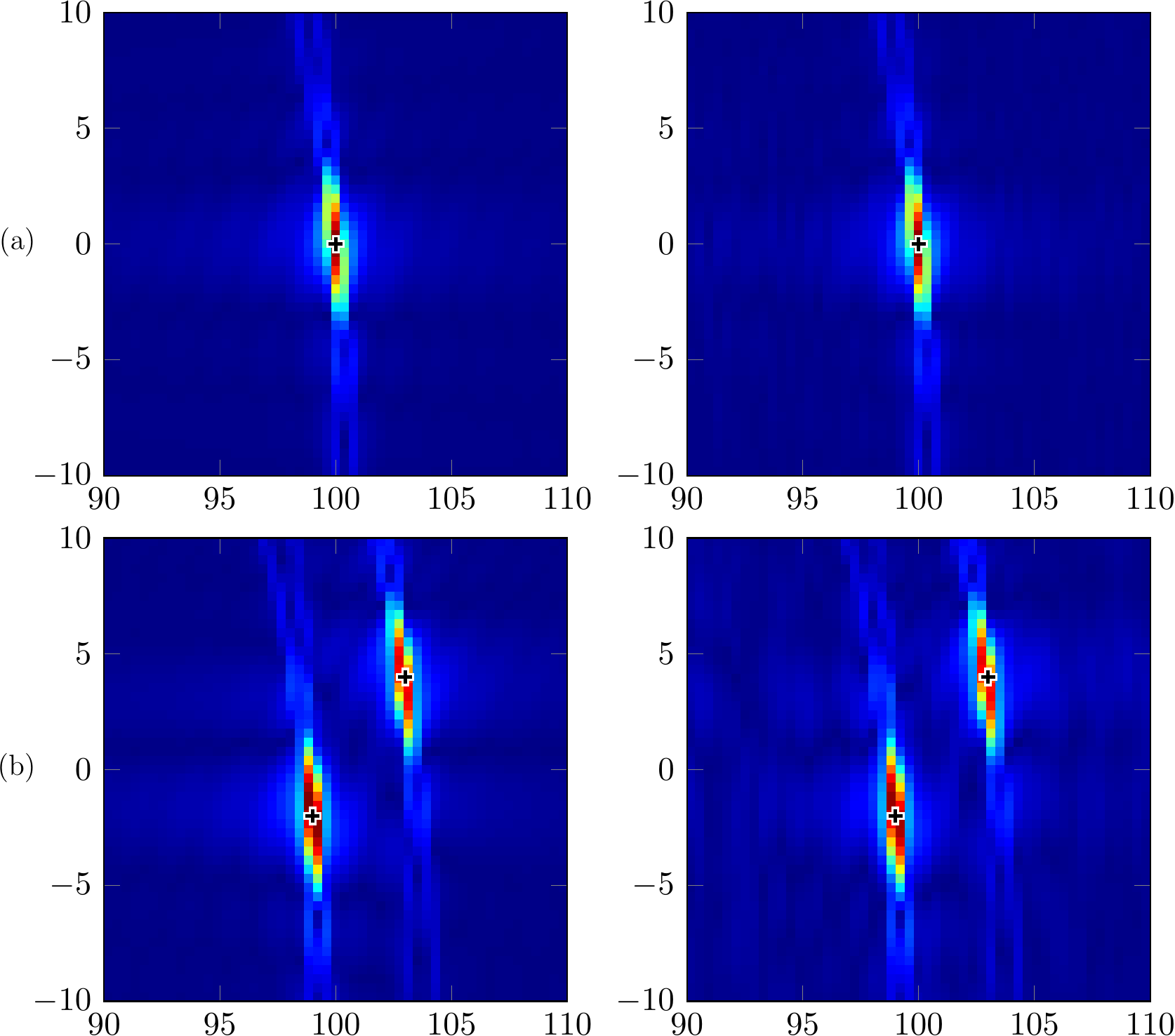}
\caption{Kirchhoff images of (a) one point and (b) two point reflectors,
whose true positions are indicated with crosses. The left column uses
the full waveform data $\vp$, while the right column use the recovered
data $\vp+\zeta\vg_0$. The horizontal
and vertical axes display the range and cross-range respectively, with
scales in central wavelengths $\lambda_0$.}\label{fig:pointref}
\end{figure}

\begin{figure}[!hbtp]
\centering
\begin{tabular}{p{0.5\textwidth}p{0.5\textwidth}}
\centering $\KM[\vp](\vec{\vy})$ &
\centering $\KM[\vp+\zeta\vg_0](\vec{\vy})$
\end{tabular}
\includegraphics[width=\textwidth]{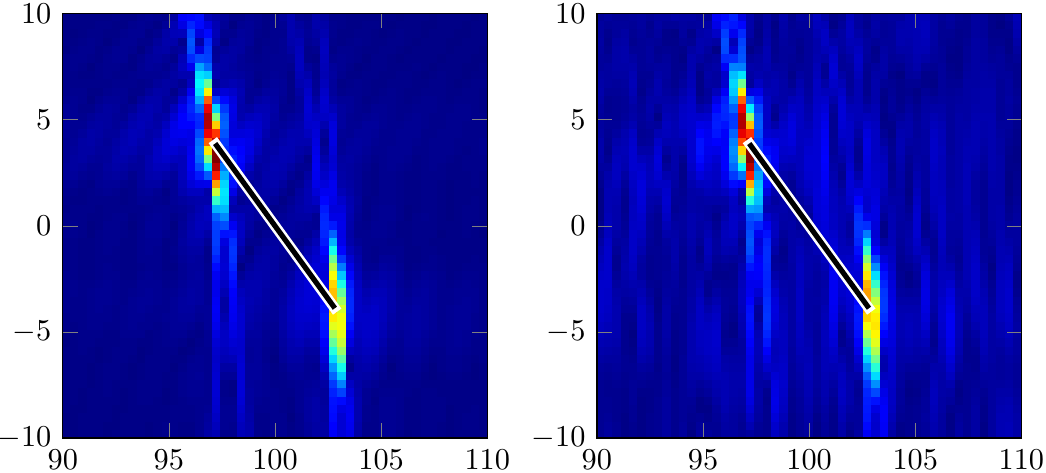}
\caption{Kirchhoff images of an extended reflector. The left column uses
the full waveform data $\vp$, while the right column use the recovered
data $\vp+\zeta\vg_0$. The horizontal
and vertical axes display the range and cross-range respectively, with
scales in central wavelengths $\lambda_0$.}\label{fig:extended_scatterer}
\end{figure}

We now demonstrate the robustness of the proposed method with respect to
additive noise (see section~\ref{sec:snr}). Here we have taken a
realization of the data for a single point reflector (c.f.
figure~\ref{fig:pointref}a) and perturbed each measurement with
additive noise as follows. The $m-$th signal 
$\hat{u}_m(\vec{\vx}_r,\omega)$ has total power $p_m = \int
|\hat{u}_m(\vec{\vx}_r,\omega)|^2 d\omega$. We construct a Gaussian
signal $\xi_m(t)$ with mean zero, spectral density \eqref{eq:noisepsd},
$l_c \approx 1.25$ $\mu$s and total power $1$. This allows to obtain the
perturbed total field $\hat{u}_m(\vec{\vx}_r,\omega)+\sqrt{\nu
p_m}\hat{\xi}_m(\omega)$ for some $\nu >0$. The $m-$th measurement with
additive noise is thus $d_m(\vec{\vx}_r,\omega) =
|\hat{u}_m(\vec{\vx}_r,\omega)|^2 + \nu p_m|\hat{\xi}(\omega)|^2$. Thus
the ratio of the signal power to the noise power is $1/\nu$. The
signal-to-noise ratio (SNR) is then \[ \text{SNR}_m = -10\log_{10}(\nu) \text{dB}.
\]
Figure~\ref{fig:centered_noise} shows the reconstruction from data with
SNR$_m=0$ dB for each $m$, meaning that the signal and the noise have
the same power.

\begin{figure}[!hbtp]
\centering
\begin{tabular}{p{0.5\textwidth}p{0.5\textwidth}}
\centering $\KM[\vp](\vec{\vy})$ &
\centering $\KM[\vp+\zeta\vg_0](\vec{\vy})$
\end{tabular}
\includegraphics[width=\textwidth]{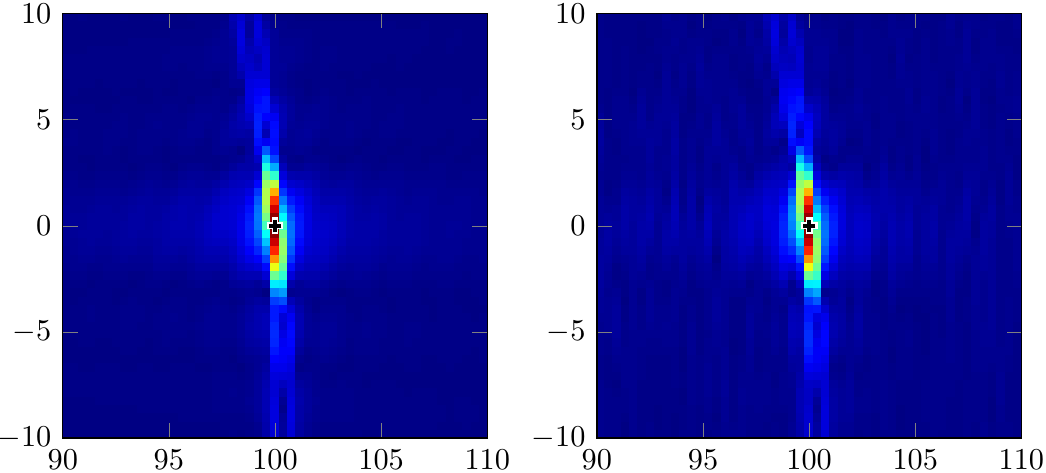}
\caption{Additive noise: (left) array response
vector migration $\KM[\vp](\vec{\vy})$, (right) recovered array response
vector migration $\KM[\vp+\zeta\vg_0](\vec{\vy})$ for SNR$_m=0$dB. The
horizontal and vertical axes display the range and cross-range
respectively measured in central wavelengths
$\lambda_0$.}\label{fig:centered_noise}
\end{figure}

Lastly we perform an experiment that sends as the $m-$th illumination
the usual correlated pair illumination $\vf_m$, and uncorrelated noise
from the remaining sources on the array $\sA$ (see
remark~\ref{rem:backgroundnoise}). To generate this illumination we
place the \emph{same} realization of the signal $\hat{f}(\omega)$ at the
locations $\vx_{i(m)}$ and $\vx_{j(m)}$, and \emph{independent}
realizations of $\hat{f}(\omega)$ at the remaining source locations.
Similarly, a reference illumination is generated by placing
\emph{independent} realizations of $\hat{f}(\omega)$ at \emph{all}
locations on the array $\sA$. By measuring the autocorrelation of the
resulting fields we obtain data that is essentially the same form as
$d_m(\vec{\vx}_r,\omega)$.  Figure~\ref{fig:backgroundnoise} shows this
experiment with the single point reflector located at
$\vec{\vy}=(100\lambda_0,0)$ and reflectivity $\rho(\vec{\vy}) =
1\times 10^{-8}$. 

\begin{figure}[!hbtp]
\centering
\begin{tabular}{p{0.5\textwidth}p{0.5\textwidth}}
\centering $\KM[\vp](\vec{\vy})$ &
\centering $\KM[\vp+\zeta\vg_0](\vec{\vy})$
\end{tabular}
\includegraphics[width=\textwidth]{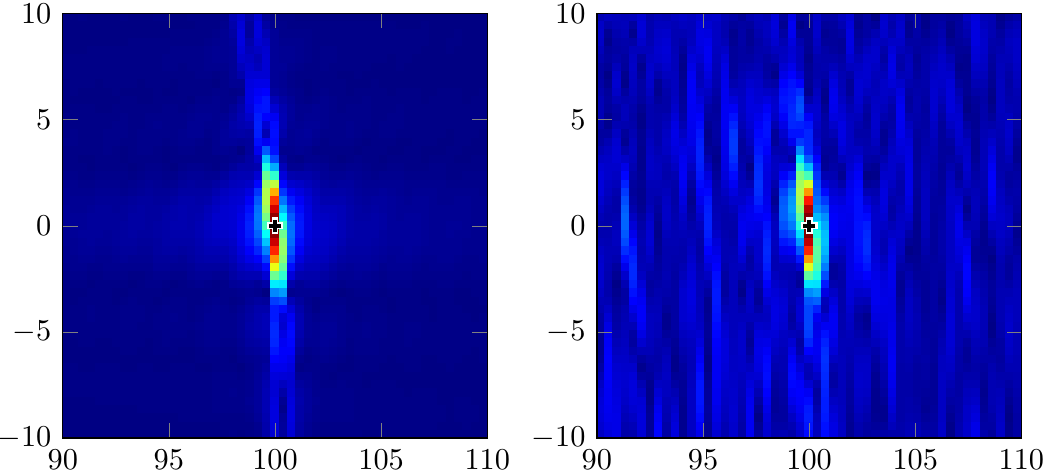}
\caption{Uncorrelated background illumination: (left) array response vector
migration $\KM[\vp](\vec{\vy})$,  (b) recovered array response vector
migration $\KM[\vp+\zeta\vg_0](\vec{\vy})$ for SNR$_m=0$dB. The
horizontal and vertical axes display the range and cross-range
respectively measured in central wavelengths
$\lambda_0$.}\label{fig:backgroundnoise} 
\end{figure}

%%%%%%%%%%%%%%%%%%%%%%%%%%%%%%%%%%%%%%%%%%%%%%%%%%%%%%%%%%%%%%%%%%%%%%%%
\subsection{Optic regime}\label{sec:numoptic} For imaging in an optic
regime, we use the background wave velocity $c_0 = 3\times 10^8$ m/s and
central frequency $\approx 589$ THz which gives a central wavelength
$\lambda_0 \approx 509$ nm. Our source array $\sA$ is again centered at the
origin, but now consists of $1001$ sources located at coordinates
$\vec{\vx}_s = (0,-500\lambda_0+(s-1)\lambda_0)$ for $s=1,\ldots,1001$,
and we set $\vec{\vx}_r = (-1000\lambda_0,-1000\lambda_0)$. 

We generate intensity data $\vd(\vec{\vx}_r,\omega)$ as 
\[
d_m(\vx_r,\omega) = 
\Big|\big(\vg_0+\vp\big)^T(\ve_{i(m)}+\ve_{j(m)})\Big|^2-\Big|\big(\vg_0+\vp\big)^T\ve_{i(m)}\Big|^2-\Big|\big(\vg_0+\vp\big)^T\ve_{j(m)}\Big|^2,
\]
for $100$ (angular) frequencies $\omega$ uniformly spaced in the
frequency band $[429,750]$ THz. This
corresponds to performing the source pair experiments (source pair
illuminations and single source reference illuminations) for $100$
different monochromatic visible light sources with wavelengths
$\lambda\in[400,700]$ nm, equally spaced in frequency. Since there are a large number
of sources in this setup ($N=1001$), we implement the strategy discussed
in \S\ref{sec:invertible} to reduce the number of source pair
experiments from $N_p=N(N-1)/2$ to $N_p = 2N$.

As before, we use the pseudoinverse $\mM^\dagger$ to recover
$\vp+\zeta\vg_0$ for each frequency $\omega\in\sB$, and then use the
Kirchhoff migration functional (\S\ref{sec:km}) to form an image. Here
we use the image window $\sW =
\{(5000\lambda_0+i\lambda_0/2.5,j\lambda_0/2.5), \text{for
$i,j=-25,\ldots,25$}\}$. In figure~\ref{fig:optic}(b)
we demonstrate the migrated image for two point reflectors placed at
$\vec{\vy}_1=(4098\lambda_0,3\lambda_0)$ and
$\vec{\vy}_2=(5004\lambda_0,-5\lambda_0)$ each with reflectivity
$\rho(\vec{\vy}_i)=1\times 10^{-17}$. Although we are significantly
undersampling the data in frequency and the source spacing is larger
than $\lambda_0/2$, the spot sizes still exhibit the Kirchhoff migration
resolution estimates (\S\ref{sec:km}) of $\lambda_0L/a\approx
5\lambda_0$ in cross-range and $c_0/|\sB|\approx 2\lambda_0$ in range. 

\begin{figure}[!hbtp]
\centering
\begin{tabular}{p{0.5\textwidth}p{0.5\textwidth}}
\centering $\KM[\vp](\vec{\vy})$ &
\centering $\KM[\vp+\zeta\vg_0](\vec{\vy})$
\end{tabular}
\includegraphics[width=\textwidth]{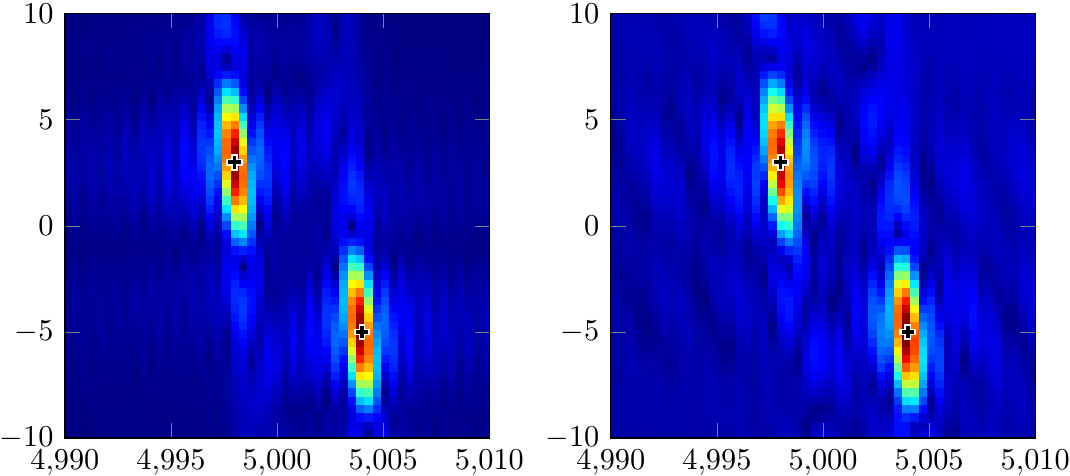}
\caption{Optic regime:(left) array response vector migration
$\KM[\vp](\vec{\vy})$,  (b) recovered array response vector migration
$\KM[\vp+\zeta\vg_0](\vec{\vy})$. The horizontal and vertical axes
display the range and cross-range respectively measured in central
wavelengths $\lambda_0$.}\label{fig:optic} 
\end{figure}

%%%%%%%%%%%%%%%%%%%%%%%%%%%%%%%%%%%%%%%%%%%%%%%%%%%%%%%%%%%%%%%%%%%%%%%%%
%%% Discussion
\section{Discussion}\label{sec:discussion}
By sending correlated signals from different pairs of locations we have
shown that from intensity data we can recover full waveform data by
solving a linear system. This linear system has a known one-dimensional
nullspace provided the sources and receiver satisfy the distance
conditions given by assumption~\ref{assump:rankcond}, which allows for
the recovery of $\vp+\zeta\vg_0$. We show this quantity is enough to
use standard migration techniques (e.g. Kirchhoff migration $\KM$)
provided the sources and receiver satisfy the additional geometric
conditions of assumption~\ref{assump:geometricimg}. Thus we obtain full
waveform resolution estimates for an image formed from intensity-only
data.

Our method relies only on knowledge of paired source locations and the
correlation of the signals being sent. This allows us to relax
illumination control by using paired stochastic signals. By measuring
autocorrelations of the resulting fields, we obtain essentially the same
intensity data as with using deterministic source pairs. These
stochastic illuminations can be created e.g. by using a configurable
mask that is parallel to the wave fronts of an incoherent plane wave. 

The linear system we solve has size $2N\times 2N$ and is very sparse (up
to 4 non-zero entries per row). In our simulations we used
$\mM^\dagger$, however sparse solvers such as CGLS (see e.g.
\cite{Hansen:1998:RDD}) could be used. To form the system we need at
least $3N$ different illuminations, $2N$ pair illuminations plus $N$
reference illuminations. However, in our illumination strategy, the
phase of the source signals does not need to be known. We replace the
direct phase control by the natural phase modulation that comes from the
different positions of the signals.

We use the geometric imaging conditions
(assumption~\ref{assump:geometricimg}) to show the nullspace of $\mM$
does not affect imaging via $\KM$. This assumption imposes some
restrictions on the juxtaposition of the sources and receiver and in
turn on the forms of illuminations we can consider. For example, using a
stationary phase argument, it can be shown the autocorrelation of the
total field is negligible if spatially continuous array illuminations
(rather than paired point sources) are used. In future work, we would
like to address this more thoroughly to determine if more general
illuminations can be used. It may also be interesting to see if the
source pair strategy we propose will work for other imaging setups. 

%%%%%%%%%%%%%%%%%%%%%%%%%%%%%%%%%%%%%%%%%%%%%%%%%%%%%%%%%%%%%%%%%%%%%%%%
%% Acknowledgements
\section*{Acknowledgements}
The authors would like to thank Alexander Mamonov, Andy Thaler and Greg
Rice for insightful discussions related to this project. PB is
especially grateful to Graeme W. Milton for his generous support. The
work of P. Bardsley was supported by the National Science Foundation
grants DMS-1411577 and DMS-1211359.  The work of F. Guevara Vasquez was
supported by the National Science Foundation grant DMS-1411577.

%%%%%%%%%%%%%%%%%%%%%%%%%%%%%%%%%%%%%%%%%%%%%%%%%%%%%%%%%%%%%%%%%%%%%%%%
%% Appendix
\appendix
\section{Proof of proposition~\ref{prop:statstab}}\label{sec:appendix}
In this appendix we prove proposition~\ref{prop:statstab} which details the
statistical stability of the measured autocorrelation
\eqref{eq:empautocorr} with respect to realizations of the illumination
$\vf$. The theorem and proof are patterned after the result by
Garnier~and~Papanicolaou~\cite[Proposition 4.1]{Garnier:2009:PSI}, only
we make small modifications to allow for complex fields and the form
\eqref{eq:noise} of the correlation function $\mR(\tau)$. 

\begin{proof}
Since we are assuming $\vf$ is a stationary process in $t$, the resulting total
field $u$ is also a stationary random process in $t$. So we have
\[
\langle \overline{u}(\vec{\vx}_r,t)u(\vec{\vx}_r,t+\tau)\rangle = \langle
\overline{u}(\vec{\vx}_r,0)u(\vec{\vx}_r,\tau)\rangle,\\
\]
which allows us to compute
\[
\begin{aligned}
\langle \C(\vec{\vx}_r,\tau)\rangle &=
\frac{1}{2T}\int_{-T}^Tdt\langle
\overline{u}(\vec{\vx}_r,t)u(\vec{\vx}_r,t+\tau)\rangle \\
&= \frac{1}{2T}\int_{-T}^T dt\langle
\overline{u}(\vec{\vx}_r,0)u(\vec{\vx}_r,\tau)\rangle=
\langle
\overline{u}(\vec{\vx}_r,0)u(\vec{\vx}_r,\tau)\rangle.
\end{aligned}
\]
So \eqref{eq:empautocorr} is independent of $T$.  By expressing the
quantity $\langle\overline{u}(\vec{\vx}_r,0)u(\vec{\vx}_r,\tau)\rangle$
through the Green's function $G$ we verify \eqref{eq:indT}:
\[
\begin{aligned}
\langle \C(\vec{\vx}_r,\tau)\rangle &=
\sum_{p,p'=1}^N\int dt'\int dt''\overline{G}(\vec{\vx}_r,\vec{\vx}_p,-t')G(\vec{\vx}_r,\vec{\vx}_{p'},\tau-t'')\langle
\overline{f}_{p}(t')f_{p'}(t'')\rangle\\
&=\sum_{p,p'=1}^N\int dt'\int dt''\overline{G}(\vec{\vx}_r,\vec{\vx}_p,-t')G(\vec{\vx}_r,\vec{\vx}_{p'},\tau-t'')R_{p,p'}(t''-t')\\
&=\frac{1}{2\pi}\int d\omega
e^{-\ii\omega\tau}\vg(\vec{\vx}_r,\omega)^*\hat{\mR}(\omega)\vg(\vec{\vx}_r,\omega).
\end{aligned}
\]

To show the ergodicity \eqref{eq:ergodicity}, we need to compute the
variance of $\C$. We first compute the covariance as
\begin{equation}\label{eq:ctcov}\begin{aligned}
\cov&\big(\C(\vec{\vx}_r,\tau),\C(\vec{\vx}_r,\tau+\Delta\tau)\big)
=
\sum_{p,p',q,q'=1}^N\frac{1}{(4T)^2}\int_{-T}^T\int_{-T}^T
dt dt'\int ds ds' du du'\\
&
\times
G(\vec{\vx}_r,\vec{\vx}_p,s)\overline{G}(\vec{\vx}_r,\vec{\vx}_{p'},u-\tau)\overline{G}(\vec{\vx}_r,\vec{\vx}_q,s')G(\vec{\vx}_r,\vec{\vx}_{q'},u'-\tau-\Delta\tau)\\
&\times\Big(\langle
f_p(t-s)\overline{f}_{p'}(t-u)\overline{f}_q(t'-s')f_{q'}(t'-u')\rangle\\
&-\langle
f_p(t-s)\overline{f}_{p'}(t-u)\rangle\langle\overline{f}_q(t'-s')f_{q'}(t'-u')\rangle\Big).
\end{aligned}
\end{equation}
The product of the second order moments is
\[
\langle
f_p(t-s)\overline{f}_{p'}(t-u)\rangle\langle\overline{f}_q(t'-s')f_{q'}(t'-u')\rangle
= R_{p',p}(u-s)R_{q,q'}(s'-u').
\]
Since $\vf(t)$ is Gaussian (in time), the fourth order moment is given by the
complex Gaussian moment theorem
(see e.g. \cite{Reed:1962:MTCG}) as
\[
\begin{aligned}
\langle
f_p(t-s)\overline{f}_{p'}(t-u)\overline{f}_q(t'-s')f_{q'}(t'-u')\rangle
&=
R_{p',p}(u-s)R_{q,q'}(s'-u')\\
&+R_{q,p}(t-t'-s+s')R_{p',q'}(t'-t-u'+u).\\
\end{aligned}
\]

We now integrate over the $t,t'$ variables to obtain
\begin{align*}
&\frac{1}{4T^2}\int_{-T}^Tdt \int_{-T}^T dt' \Big(\langle
f_p(t-s)\overline{f}_{p'}(t-u)\overline{f}_q(t'-s')f_{q'}(t'-u')\rangle\\
&\qquad\qquad\qquad-\langle
f_p(t-s)\overline{f}_{p'}(t-u)\rangle\langle\overline{f}_q(t'-s')f_{q'}(t'-u')\rangle\Big)\\
&=\frac{1}{4T^2}\int_{-T}^Tdt \int_{-T}^T
dt'R_{q,p}(t-t'-s+s')R_{p',q'}(t'-t-u'+u)\\
&= \int d\omega \int d\omega' \sinc^2\big((\omega-\omega')T\big)
e^{\ii\omega'(s-s')}e^{-\ii\omega(u-u')}\hat{R}_{q,p}(\omega)\hat{R}_{p',q'}(\omega').\\
%%
%&=4\pi^2 S(s-s',u-u')H_m(\vec{\vx}_p,\vec{\vx}_q)H_m(\vec{\vx}_{p'},\vec{\vx}_{q'}).
\end{align*}

Plugging this into \eqref{eq:ctcov} we obtain
\[
\begin{aligned}
&\cov\big(\C(\vec{\vx}_r,\tau),\C(\vec{\vx}_r,\tau+\Delta\tau)\big)= 
\sum_{p,p',q,q'=1}^N\int
d\omega \int d\omega'
\sinc^2\big((\omega-\omega')T\big)\\
&\qquad\times
\hat{G}(\vec{\vx}_r,\vec{\vx}_p,\omega')\overline{\hat{G}}(\vec{\vx}_r,\vec{\vx}_{p'},\omega)\overline{\hat{G}}(\vec{\vx}_r,\vec{\vx}_q,\omega)\hat{G}(\vec{\vx}_r,\vec{\vx}_{q'},\omega)\hat{R}_{q,p}(\omega)\hat{R}_{p',q'}(\omega')e^{\ii\omega\Delta\tau}\\
&= \int d\omega \int d\omega'
\sinc^2\big((\omega-\omega')T\big)\big(\vg(\vec{\vx}_r,\omega)^*\hat{\mR}(\omega)\vg(\vec{\vx}_r,\omega')\big)\\
&\qquad\qquad\qquad\qquad\qquad\qquad\times\big(\vg(\vec{\vx}_r,\omega)^*\hat{\mR}(\omega')\vg(\vec{\vx}_r,\omega')\big)e^{\ii\omega\Delta\tau},
\end{aligned}
\] 
where $\vg=\vg_0+\vp$ is given by \eqref{eq:homoggreen} and
\eqref{eq:arv}, and $\big(\hat{\mR}(\omega)\big)_{i,j} =
\hat{R}_{i,j}(\omega)$ is a $\complex^{N\times N}$
Hermitian matrix for each $\omega$. Then taking $T\to\infty$  we compute the variance as
\[
T\var\big(\C_m(\vec{\vx}_r,\tau)\big)\xrightarrow{T\to\infty}\int
d\omega
\Big|\vg(\vec{\vx}_r,\omega)^*\hat{\mR}(\omega)\vg(\vec{\vx}_r,\omega)\Big|^2,
\]
and so the variance is $\bigO(1/T)$ as $T\to\infty$. This establishes
\eqref{eq:ergodicity}.
\end{proof}

%%%%%%%%%%%%%%%%%%%%%%%%%%%%%%%%%%%%%%%%%%%%%%%%%%%%%%%%%%%%%%%%%%%%%%%%
%% Bibliography
\bibliographystyle{siam}
\bibliography{bibautocorr}

\end{document}